\documentclass[11pt]{article}

\usepackage{amsmath, amsfonts, amssymb, amsthm,  graphicx, enumerate}
\usepackage[letterpaper, left=1truein, right=1truein, top = 1truein, bottom = 1truein]{geometry}

\usepackage[authoryear, round]{natbib}

\usepackage[
            CJKbookmarks=true,
            bookmarksnumbered=true,
            bookmarksopen=true,
            colorlinks=true,
            citecolor=brown,
            linkcolor=blue,
            anchorcolor=red,
            urlcolor=blue
            ]{hyperref}

\usepackage[usenames]{color}

\usepackage[ruled, vlined, lined, commentsnumbered]{algorithm2e}

\usepackage{prettyref,soul,xspace}

\newcommand{\eg}{e.g.\xspace}
\newcommand{\ie}{i.e.\xspace}
\newcommand{\iid}{i.i.d.\xspace}

\newcommand{\ones}{\mathbf 1}
\newcommand{\reals}{{\mathbb{R}}}
\newcommand{\integers}{{\mathbb{Z}}}

\newcommand{\supp}{{\rm supp}}
\newcommand{\eexp}{{\rm e}}

\newcommand{\diff}{{\rm d}}

\newcommand{\rank}{\mathop{\sf rank}}

\newcommand{\diag}{\mathop{\text{diag}}}

\newcommand{\Expect}{\mathbb{E}}

\newcommand{\Prob}{\mathbb{P}}

\newcommand{\argmin}{\mathop{\rm argmin}}

\newcommand{\Th}{{^{\rm th}}}

\newrefformat{eq}{(\ref{#1})}
\newrefformat{chap}{Chapter~\ref{#1}}
\newrefformat{sec}{Section~\ref{#1}}
\newrefformat{algo}{Algorithm~\ref{#1}}
\newrefformat{fig}{Fig.~\ref{#1}}
\newrefformat{tab}{Table~\ref{#1}}
\newrefformat{rmk}{Remark~\ref{#1}}
\newrefformat{clm}{Claim~\ref{#1}}
\newrefformat{def}{Definition~\ref{#1}}
\newrefformat{cor}{Corollary~\ref{#1}}
\newrefformat{lmm}{Lemma~\ref{#1}}
\newrefformat{lemma}{Lemma~\ref{#1}}
\newrefformat{prop}{Proposition~\ref{#1}}
\newrefformat{app}{Appendix~\ref{#1}}
\newrefformat{ex}{Example~\ref{#1}}
\newrefformat{exer}{Exercise~\ref{#1}}
\newrefformat{soln}{Solution~\ref{#1}}
\newrefformat{cond}{Condition~\ref{#1}}

\newcommand{\lunder}[1]{{\underset{\raise0.3em\hbox{$\smash{\scriptscriptstyle-}$}}{#1}}}

\newcommand{\norm}[1]{\left\|{#1} \right\|}

\newcommand{\lnorm}[2]{\left\|{#1} \right\|_{{#2}}}

\newcommand{\Fnorm}[1]{\lnorm{#1}{\rm F}}
\newcommand{\fnorm}[1]{\|#1\|_{\rm F}}

\newcommand{\opnorm}[1]{\|#1\|_{\rm op}}

\def\innergetnumber#1[#2]#3{#2}
\def\getnumber{\expandafter\innergetnumber\jobname}

\newcommand{\bszero}{{\boldsymbol{0}}}

\newcommand{\calB}{{\mathcal{B}}}

\newcommand{\calF}{{\mathcal{F}}}

\newcommand{\calM}{{\mathcal{M}}}

\newcommand{\comp}[1]{{#1^{\rm c}}}

\newcommand{\pth}[1]{\left( #1 \right)}
\newcommand{\qth}[1]{\left[ #1 \right]}
\newcommand{\sth}[1]{\left\{ #1 \right\}}

\newcommand{\KL}[2]{D(#1 \, || \, #2)}

\usepackage{amsmath}
\usepackage{latexsym}
\usepackage{mathrsfs} 
\usepackage{amsfonts}
\usepackage{amssymb}

\newcommand{\bel}{\begin{eqnarray}\label}
\newcommand{\eel}{\end{eqnarray}}
\newcommand{\bes}{\begin{eqnarray*}}
\newcommand{\ees}{\end{eqnarray*}}

\def\benu{\begin{enumerate}}
\def\eenu{\end{enumerate}}

\def\argmin{\mathop{\rm arg\, min}}

\def\complex{\mathop{{\rm I}\kern-.58em\hbox{\rm C}}\nolimits}

\def\diag{\hbox{diag}}

\def\supp{\hbox{supp}}

\newcommand{\R}{{\mathbb{R}}}

\def\Ahat{\widehat{A}}

\def\calB{{\cal B}}\def\Bhat{\widehat{B}}

\def\calF{{\cal F}}

\def\calM{{\cal M}}

\def\rhat{\widehat{r}}

\def\Yhat{\widehat{Y}}\def\Ytil{{\widetilde Y}}

\def\Ztil{{\widetilde Z}}

\def\tkappa{\widetilde{\kappa}}

\def\lam{\lambda}

\def\trho{\widetilde{\rho}}

\def\sigmahat{\widehat{\sigma}}\def\hsigma{\widehat{\sigma}}
\def\tsigma{\widetilde{\sigma}}

\usepackage{tikz}
\usepackage{pgfplots,makecell}

\newcommand{\col}[2]{{#1}_{* {#2}}}
\newcommand{\row}[2]{{#1}_{{#2} *}}
\newcommand{\sqnorm}[1]{\| #1 \|_{{\rm s}_q}}

\newcommand{\wh}{\widehat}
\newcommand{\wt}{\widetilde}

\theoremstyle{plain}
\newtheorem{theorem}{Theorem}
\newtheorem*{theorem*}{Theorem}

\newtheorem{proposition}{Proposition}
\newtheorem{lemma}{Lemma}

\theoremstyle{definition}
\newtheorem{definition}{Definition}
\newtheorem{condition}{Condition}

\newtheorem{remark}{Remark}

\title{Adaptive Estimation in Two-way Sparse Reduced-rank 
Regression\footnote{
The authors would like to thank Dr. Kun Chen for kindly sharing his code on the exclusive extraction algorithm which we have used in \prettyref{sec:num} of the paper.
An earlier version of the present paper \citep{MaSun14} under the title ``Adaptive sparse reduced-rank regression'' studied a one-way sparse reduced-rank regression model, which can be viewed as a special case of the model considered in this paper.
The earlier version has been 
uploaded on arXiv, but is not intended for publication.} 
}

\author{
Zhuang Ma\thanks{Department of Statistics,
University of Pennsylvania. Email: {\ttfamily zhuangma@wharton.upenn.edu}.}, ~~~
Zongming Ma\thanks{Department of Statistics,
University of Pennsylvania. Email: {\ttfamily zongming@wharton.upenn.edu}.
}
~~~and~~~Tingni Sun\thanks{
Department of Mathematics, University of Maryland, College Park.
Email: {\ttfamily tingni@math.umd.edu}. 
}
\\
}

\date{ }

\begin{document}
\maketitle

\begin{abstract}
This paper studies the problem of estimating a large coefficient matrix in a multiple response linear regression model when the coefficient matrix could be both of low rank and sparse in the sense that most nonzero entries concentrate on a few rows and columns.
We are especially interested in the high dimensional settings where the number of predictors and/or response variables can be much larger than the number of observations.
We propose a new estimation scheme, which achieves competitive numerical performance
and at the same time allows fast computation.
Moreover, we show that (a slight variant of) the proposed estimator achieves near optimal non-asymptotic minimax rates of estimation under a collection of squared Schatten norm losses simultaneously by providing both the error bounds for the estimator and minimax lower bounds. 
The effectiveness of the proposed algorithm is also demonstrated on an \textit{in vivo} calcium imaging dataset.

\bigskip

\end{abstract}

\section{Introduction}
\label{sec:intro}

High dimensional sparse linear regression has been one of the central topics of high dimensional statistical inference. 
When the response is univariate, researchers have developed a dazzling collection of tools to take advantage of the potential sparsity of the regression coefficients, \eg, Lasso \citep{Tibs96,Chen98}, SCAD \citep{FanLi01}, Dantzig selector \citep{Candes07}, MCP \citep{Zhang10}, etc.
In contemporary applications, we routinely face multivariate or even high dimensional response variables together with a large number of predictors, while the sample size can be much smaller.
For example, in a cognitive neuroscience study, \citet{Vounou12} used around ten thousand voxels from fMRI imaging as the response variables for each subject, and over four hundred thousand SNPs (single-nucleotide polymorphisms) as predictors. In comparison, the sample size was just several hundred.

Let $n$ denote the sample size, $m$ the number of responses, and $p$ the number of predictors.
We observe a pair of matrices $Y$ and $X$ from the following linear model
\begin{equation}
\label{eq:model}
Y = XA + Z,
\end{equation}
where $Y$ is an $n \times m$ response matrix, $X$ is an $n \times p$ design matrix, $A$ is a $p \times m$ coefficient matrix that we are interested in estimating, and $Z$ is an unobserved $n \times m$ matrix with \iid~noise entries.
Thus, the $i\Th$ rows of $Y$ and $X$ collect the measurements of the response and the predictor variables on the $i\Th$ subject, respectively.
When either the number of predictors $p$ or the number of response variables $m$ is large, it is hard to estimate the coefficient matrix $A$ accurately unless certain structural assumption is imposed so that its intrinsic dimension is low.

In the literature, researchers have considered several important types of structural assumptions.
One is \emph{low-rankness} where the rank of $A$ is assumed to be much smaller than its matrix dimensions $p$ and $m$. 
Model \eqref{eq:model}
with such a structure has been referred to as reduced-rank regression and has been widely used in econometrics. 
See, for instance, \citet{Izenman75}, \citet{Reinsel98} and the references therein.
The other is \emph{sparsity} where a large number of entries in the coefficient matrix are zeros. 
One may consider several different types of sparsity depending on the application problem one has in mind.
If only $s$ out of the $p$ rows in $A$ have non-zero entries, it is called \emph{row sparsity}. In other words, only a small subset of size $s$ out of the $p$ predictors contribute to the variation of $Y$. Structures of this kind arise naturally in the context of multi-task learning \citep{Koltchinskii11}. 
It can also be viewed as a leading example of \emph{group sparsity} \citep{Yuan06}, where the rows of $A$ form natural groups. 
If only $k$ out of the $m$ columns in $A$ have non-zero entries, it is called \emph{column sparsity}. 
In this case, only $k$ out of the $m$ response variables are affected by the predictors under consideration. 

In this paper, we are interested in the situation where low-rankness, row sparsity and column sparsity could be present in the coefficient matrix simultaneously.
In what follows, we refer to model \eqref{eq:model} with these structures as the \emph{two-way sparse reduced-rank regression} model. 
The interest in such a model
comes from both applications and theory, and has risen significantly in recent years.
In applications such as genomics and neurosciences, researchers can now measure a lot of response and predictor variables and so the size of the coefficient matrix is ever increasing. 
Thus, imposing both low-rankness and two-way sparsity leads to enhanced interpretability and hence can be more attractive than simply imposing one type of structure.
For instance, \cite{ma2014learning} conducted a case study of regulatory relationships between different genome-wide measurements, in which the predictors are micro-RNA measurements and the response variables are gene expression levels.
The sparsity results from the fact that a relatively small number of micro-RNAs regulated a small collection of genes under the specific experiments of interest, and the low-rankness assumption is reasonable since only a handful of regulatory programs were present. 
For estimating the coefficient matrix in this model, several algorithms have been introduced. 
See, for instance, \cite{ChenChan12} and \cite{ma2014learning}.
However, to the best of our limited knowledge, there is no theoretical guarantee on the performance of these procedures in the high dimensional regime where the number of predictors and/or response variables exceeds the sample size.
%

\paragraph{Main contributions}
The main contributions of the present paper are two-folded. 
On one hand, we propose a new computationally efficient estimator for the coefficient matrix in \eqref{eq:model} that could take advantage of the potential presence of low-rankness and sparsity adaptively.
The new estimator shows competitive numerical performance under a variety of simulation settings when compared with state-of-the-art methods.
We also demonstrate how the estimation scheme can play a critical role in analyzing the spatial-temporal structure in calcium imaging data.
On the other hand, we obtain new minimax estimation rates of the coefficient matrix with respect to a large class of squared Schatten norm losses and show that (a slight variant of) our estimator can achieve the near optimal rates adaptively for this large collection of loss functions simultaneously when the noise terms are homoscedastic and Gaussian.

\paragraph{Connection to the literature}
When the coefficient matrix is \emph{either} sparse \emph{or} of low rank, researchers have obtained deep understanding on how the optimal mean squared estimation/prediction error depends on the model parameters and on how to achieve near optimal error rates without knowing the true rank or sparsity. 
See, for instance, \citet{Bunea11rank} for the low rank case, and \citet{Huang10} and \citet{Lounici11} for the row sparse case.



In addition, researchers have performed extensive study of the case where both low-rankness and row sparsity are present.
\citet{Chen12} proposed a weighted rank-constrained group Lasso approach with two heuristic numerical algorithms and studied its fixed dimension large sample asymptotics.
\citet{Bunea12} derived oracle inequalities and studied the minimax rates under squared prediction error loss for this model in the high dimensional setting. 
See also 
\citet{she2014selectable} and 
an earlier version of the present paper \citep{MaSun14}.

The line of work that is closest to the present paper includes two recent papers: \cite{ChenChan12} and \cite{ma2014learning}. 
The main focus of these two papers was on methodology.
In comparison, the present paper not only proposes a new method but also justifies its practical effectiveness by both numerical and theoretical studies.
From a slightly different perspective, a series of papers have considered the problem of sparse SVD \citep{lee10,Yang11,yang2014rate}, which can be viewed as a special case of two-way sparse reduced-rank regression with orthogonal design.


\paragraph{Organization} 
The rest of the paper is organized as follows.
\prettyref{sec:method} presents our new methodology for 
obtaining a simultaneously sparse and low rank estimator of
the coefficient matrix.
Its competitive numerical performance is demonstrated in \prettyref{sec:num} through both simulated and real data examples.
In \prettyref{sec:theory}, we provide finite sample upper bounds for (a slight variant of) the proposed estimator with respect to a collection of squared Schatten norm losses.
In addition, we derive minimax lower bounds and hence show that the proposed estimator is simultaneously adaptive and near optimal with respect to all loss functions under consideration.
\prettyref{sec:conclusion} discusses interesting related problems for future research.
The proofs of the theorems are presented in \prettyref{sec:proof}.

\paragraph{Notation} 
For an $n\times p$ matrix $X = (x_{ij})$, 
the $i\Th$ row of $X$ is denoted by $\row{X}{i}$ and the $j\Th$ column by $\col{X}{j}$.
For a positive integer $k$, $[k]$ denotes the index set $\{1, 2, ..., k\}$. For any set $I$, $|I|$ denotes its cardinality and $\comp{I}$ its complement. 
For two subsets $I$ and $J$ of indices, we write $X_{IJ}$ for the $|I|\times |J|$ submatrices formed by $x_{ij}$ with $(i,j) \in I \times J$.
For conciseness, we let
$\row{X}{I} = X_{I [p]}$ and $\col{X}{J} = X_{[n] J}$.
For any matrix $X$, $\supp(X)$ stands for the index set of its nonzero rows.
We denote the rank of $X$ by $\rank(X)$,
and $\sigma_i(X)$ stands for its $i\Th$ largest singular value.
For any $q\in [1,\infty)$, the Schatten-$q$ norm of $X$ is
\begin{align*}
\sqnorm{X} = \pth{\sum_{i=1}^{n\wedge p} \sigma_i^q(X)}^{1/q},
\end{align*}
and for $q = \infty$, $\|X\|_{{\rm S}_\infty} = \sigma_1(X)$.
Note that $\|X\|_{{\rm S}_2} = \Fnorm{X}$ is the Frobenius norm and $\|X\|_{{\rm S}_\infty} = \opnorm{X}$ is the operator norm of $X$.
For any vector $a$, $\norm{a}$ denotes its $\ell_2$ norm.
The $\ell_2/\ell_1$ norm of $X$ is defined as the $\ell_1$ norm of the vector consisting of its row $\ell_2$ norms: 
$\|X\|_{2,1}=\sum_{j=1}^n \|X_{j*}\|$.
If $n\geq p$ and $X$ has orthonormal columns, then we say $X$ is an orthonormal matrix, and we write $X\in O(n,p)$.
We use $\ones_d$ to denote the all-one vector in $\reals^d$.
For any real number $a$ and $b$, set $a\vee b = \max\{a,b\}$, $a\wedge b = \min\{a,b\}$ and $a_+ = a\vee 0$.

\section{Methodology}
\label{sec:method}

\subsection{Main Algorithm}

The proposed estimation scheme, called {\it Double Projected Penalization} (DPP), is summarized in \prettyref{algo:est}.
To initialize the algorithm, we need to specify the rank $r$ of the estimated coefficient matrix and a penalty function $\rho(\cdot\,;\lambda)$ to be used in group penalized regression. 
In what follows, we explain the main ideas underlying the algorithm, while
the choice of penalty and other initialization details 
are deferred to
Sections \ref{sec:group-reg} and \ref{sec:init-alg}.

The algorithm consists of two stages. The first stage involves steps 1--2 and the second stage steps 3--5.
In either stage, one first screens the columns of $Y$, then computes the $r$ leading right singular vectors of the screened response matrix, and finally performs a group penalized regression on the projected data where the projection is onto the subspace spanned by the leading right singular vectors.
The purpose of the screening step is to pick those response variables the signals of which stand out of noise.
To motivate the projection step, we observe that if the right singular vector matrix $V$ of $XA$ were known, then one could immediately reduce dimensionality by considering the new regression problem which replaces $Y$ and $A$ in \prettyref{eq:model} with their projected counterparts $YV$ and $AV$.
Thus, in either stage, we first estimate $V$ by the $r$ leading right singular vectors of the screened response matrix (a further projection is involved in the second stage), and then project the data by post-multiplying the response matrix with the estimated right singular vector matrix.
When regressing the projected responses on $X$, we actually estimate $AV$. 
Note that if $A$ has at most $s$ nonzero rows, so does $AV$. 
Thus, the rows of $AV$ form natural groups and it makes sense to induce row sparsity in our estimator of $AV$ by performing a group penalized regression.

We now move on to discuss the necessity of the second stage. 
Comparing the two stages, we note that both the screening step and the estimation of the right singular matrix $V$ are different, but both differences are due to the involvement of the matrix $U_{(1)}$.
By definition, 
$U_{(1)}\in \reals^{n\times r}$ consists of the left singular vectors of $XB_{(1)}$.
Since $B_{(1)}$ is an estimate of $AV$, the column subspace of $U_{(1)}$ estimates the left singular subspace of $XAV$, or equivalently, 
the left singular subspace of $XA$.
By projecting onto $U_{(1)}$, we increase the signal-to-noise ratio in the screening step. 
As a result, we would be able to select more columns the signals of which might have been drowned in noise in the first stage.
The inclusion of more signal columns of $Y$ would in turn contribute to the estimation accuracy of the final estimator.
Similarly, by pre-multiplying $\wt{Y}^{(1)}$ with $U_{(1)}U_{(1)}'$, we further boost the signal-to-noise ratio when estimating the right singular vector matrix $V$, and thus obtain a better estimator $V_{(1)}$.
As to be revealed by later analysis, the second stage is critical for achieving high estimation accuracy for $A$.

\begin{algorithm}[!bth]
	\SetAlgoLined
\caption{Estimation scheme for $A$ via the Double Projected Penalization
\label{algo:est}}

\KwIn{Observed response matrix $Y$, design matrix $X$, 
rank $r$, noise level $\sigma$,
positive constants $\alpha, \beta$
 and penalty function $\rho(\cdot; \lam)$ with penalty level $\lam$.}

\KwOut{Estimated coefficient matrix $\Ahat$.}


\nl Column screening of $Y$. Select columns
\bes
J_{(0)} = \Big\{j: \|Y_{*j}\|^2\ge \sigma^2(n+\alpha \sqrt{n\log (p\vee m)}) \Big\}.
\ees
Define $\Ytil^{(0)}$, where $\Ytil_{*j}^{(0)}=Y_{*j}I{\{j\in J_{(0)}\}}$.

Compute the right singular vectors of $\Ytil^{(0)}$, denoted by an $m\times r$ matrix $V_{(0)}$.

\nl Group penalized regression
\bes
&& B_{(1)}=\argmin_{B\in \R^{p\times r}} \Big\{ \|YV_{(0)}-XB\|_F^2/2+\rho(B;\lam)\Big\}, 
\ees

\nl Column screening of $Y$. Compute the left singular vectors of  $XB^{(1)}$, denoted by an $n\times r$ matrix $U_{(1)}$.
Select columns
\bes
J_{(1)} = J_{(0)} \cup \Big\{j: \|U_{(1)}'Y_{*j}\|^2\ge \beta \sigma^2(r+ 2\sqrt {3r\log  (p\vee m)}+6\log  (p\vee m)) \Big\}.
\ees
Define $\Ytil^{(1)}$, where $\Ytil_{*j}^{(1)}=Y_{*j}I{\{j\in J_{(1)}\}}$. 

Compute  the first $r$ right singular vectors of $U_{(1)}{U_{(1)}}'\Ytil^{(1)} $, denoted by an $m\times r$ matrix  $V_{(1)}$.

\nl Group penalized regression
\bes
&& B_{(2)}=\argmin_{B\in \R^{p\times r}} \Big\{ \|YV_{(1)}-XB\|_F^2/2+\rho(B;\lam)\Big\}, 
\ees

\nl Compute the estimated coefficient matrix by 
$\Ahat=B_{(2)}{V_{(1)}}'$.

\end{algorithm}

\subsection{Group Penalized Regression}
\label{sec:group-reg}
The penalized regression in steps 2 and 4 of \prettyref{algo:est} can be viewed as a special case of linear regression with group sparsity,
where each row of the coefficient matrix is considered as a group and all groups are of the same size $r$. 

Penalized regression with group structure has been extensively studied. One of the most popular procedures is the group Lasso \citep{Bakin99, Yuan06},
where the penalty function is defined by the $\ell_2/\ell_1$ matrix norm as follows
\bel{eq:penlasso}
\rho(B;\lam)=\lam\|B\|_{2,1}=\lam\sum_{j=1}^p \|B_{j*}\|_2.
\eel
The theoretical properties of group Lasso have been studied in the literature, using ideas originating from the study of Lasso. 
\citet{Huang10} showed the upper bounds for the estimation and prediction errors of group Lasso with proper penalty level under strong group sparsity and group sparse eigenvalue conditions. 
\citet{Lounici11} provided similar error bounds under a group version of the restricted eigenvalue condition. 

In \prettyref{sec:theory}, we will present a theoretically justified choice of the penalty level $\lambda$ for the group Lasso penalty function \eqref{eq:penlasso} when we have i.i.d.~Gaussian noises. 

\subsection{Initialization}
\label{sec:init-alg}

We now discuss the initialization of \prettyref{algo:est}. Throughout, we assume the noise standard deviation $\sigma$ is known. 
Otherwise, we can estimate it by
\begin{align}
\label{eq:sigmahat}
\sigmahat=\text{median}(\sigma(Y))/\sqrt{n\vee m},
\end{align}
where $\sigma(Y)$ is the collection of all nonzero singular values of $Y$.
If the true rank of $A$ is not known, we propose to
apply the estimator in \citet{Bunea11rank}, which is summarized in  \prettyref{algo:rank}.
The user specified parameter can be selected as
\begin{equation}
\label{eq:eta}
\eta = \sqrt{2m}+\sqrt{2(n\wedge p)},
\end{equation}
which was suggested by \citet{Bunea12} for Gaussian data.

\begin{algorithm}[!bt]
\caption{Rank Estimation \label{algo:rank}}
\KwIn{Response matrix $Y$, design matrix $X$, noise level $\sigma$ and a threshold level $\eta$.}
\KwOut{Estimated rank $\rhat$, initial matrix $V_{(0)}$.}

\nl Compute $P=XM^-X'$, where $M=X'X$ and $M^-$ its Moore--Penrose pseudo-inverse. 

\nl Compute the singular values of $PY$ and select 
\bes
\rhat =  \max\sth{j: \sigma_j(PY)\ge \sigma\eta }.
\ees
\end{algorithm}

In practice, we may also select the rank based on cross validation. 
Suppose the data is split into training and test samples. 
For any given value of $r\in [m\wedge p]$, we may run \prettyref{algo:est} using only the training sample, and the resulting $\wh{A}$ is then used to calculate the prediction error on the test sample.
Thus, we can select the value of $r$ that leads to the smallest prediction error on the test sample, or the smallest average prediction error if $k$-fold cross validation is used.

\section{Numerical Study}
\label{sec:num}



\subsection{Simulation} 
In this part, we compare the proposed DPP method, i.e. \prettyref{algo:est}, with the thresholding SVD method (TSVD) in \cite{ma2014learning} and the exclusive extraction algorithm (EEA) in \cite{ChenChan12}. 
For fair comparison, 
equations \prettyref{eq:sigmahat}--\eqref{eq:eta} and \prettyref{algo:rank} were applied to estimate the noise variance and the rank of the coefficient matrix for all methods in all simulation settings.  

\paragraph{Comparison under different model parameters}
We first compare these methods under different design matrices, ranks and sparsity levels.
To this end, 
we borrow several simulation settings from \citet{Bunea12}, but also add columns of pure noises in the response matrices to induce two-way sparsity. 
The rows of the design matrix $X$ are i.i.d.~random vectors sampled from a multivariate Gaussian distribution with mean zero and covariance matrix $\Sigma$, where $\Sigma_{ij}=\rho^{|i-j|}$. 
The coefficient matrix $A\in \reals^{p\times m}$ has the form
\bes
A=\begin{pmatrix}
A_1 & 0 \\ 0 & 0
\end{pmatrix}
=\begin{pmatrix}
bB_0B_1 & 0 \\ 0 & 0
\end{pmatrix}
\ees
with $b>0$, $B_0\in \reals^{s\times r}$ and $B_1\in \reals^{r\times k}$, where all entries in $B_0$ and $B_1$ are filled with i.i.d.~random numbers from $N(0,1)$. 
The noise matrix $Z\in \reals^{n\times m}$ has i.i.d.~$N(0,\sigma^2)$ entries.
The following settings are considered with $\sigma=1$ and $\rho=0.1$ or 0.9:
\begin{itemize}
\item $n=30$, $m=50$, $p=100$, $s=15$, $k=10$, $r=2$, $b=0.5$ or 1; 
\item $n=100$, $m=50$, $p=25$, $s=15$, $k=25$, $r=5$, $b=0.2$ or 0.4.
\end{itemize}
Large values of $b$ correspond to large signal-to-noise ratios.


We compare the following five estimators derived from the three methods. 
The first two estimators are computed by \prettyref{algo:est} with $\alpha=2\sqrt{3}$, $\beta=1$ and two possible choices of penalty level $\lam$. 
The one with an estimated universal penalty level $\lam_{\mathrm{univ}}=\hsigma \sqrt{2\log(p)/n}$ is denoted by DPP, while the estimator DPP.cv selects a penalty level $\lam$ from the set $\{2^{i/2}\lam_{\mathrm{univ}}: i=-5,\dots,4\}$ via 5-fold cross validation. 
The third is the TSVD estimator which
was implemented by the R package ``tsvd'' (version 1.3) with the default penalization option ``BICtype=2''. 
The last two are EEA and its iterative extension, denoted by iEEA.  

\begin{table}[!tb]
\caption{\label{tab:simu1} Performance of five methods: means and standard deviations of prediction errors, estimation errors and sizes of selected models across 50 replications.
Simulation setting 1: $n=30$, $m=50$, $p=100$, $s=|\supp(A)|=15$, $k=|\supp(A')|=10$, $r=2$.}
\smallskip
\centering
\begin{tabular}{cccccc}
\hline\hline
$b$ & Method 
& $\|\Yhat-Y\|^2_F/(mn)$ & $\|\Ahat-A\|^2_F/(mp)$ & $|\supp(\Ahat)|$ & $|\supp(\Ahat')|$ \\\hline\hline
\multicolumn{6}{c}{$\rho=0.1$}\\\hline
0.5 &  \multicolumn{5}{c}{$ \rhat =  1.92 \pm 0.27  $ }\\\hline
& DPP & $ 1.4554 \pm 0.2803  $  & $ 0.0048 \pm 0.0027  $  & $ 28.34 \pm 3.97  $  & $ 8.64 \pm 1.24  $  \\
& DPP.cv & $ 1.4586 \pm 0.2881  $  & $ 0.0048 \pm 0.0029  $  & $ 34.70 \pm 7.65  $  & $ 8.64 \pm 1.24  $  \\
& TSVD & $ 1.9856 \pm 0.4663  $  & $ 0.0099 \pm 0.0051  $  & $ 5.68 \pm 2.97  $  & $ 9.06 \pm 2.66  $  \\
& EEA & $ 1.5832 \pm 0.3131  $  & $ 0.0062 \pm 0.0031  $  & $ 29.84 \pm 6.01  $  & $ 9.34 \pm 0.96  $  \\
& iEEA & $ 1.5549 \pm 0.3184  $  & $ 0.0059 \pm 0.0030  $  & $ 20.94 \pm 4.37  $  & $ 9.26 \pm 0.96  $  \\\hline

1 &  \multicolumn{5}{c}{$ \rhat =  2 \pm 0  $ }\\\hline
& DPP & $ 2.4490 \pm 0.9516  $  & $ 0.0148 \pm 0.0087  $  & $ 33.44 \pm 3.57  $  & $ 9.60 \pm 0.57  $  \\
& DPP.cv & $ 2.4364 \pm 0.9682  $  & $ 0.0148 \pm 0.0092  $  & $ 37.24 \pm 6.65  $  & $ 9.60 \pm 0.57  $  \\
& TSVD & $ 4.4173 \pm 1.5142  $  & $ 0.0351 \pm 0.0146  $  & $ 5.74 \pm 2.65  $  & $ 10.48 \pm 2.56  $  \\
& EEA & $ 3.0317 \pm 1.2327  $  & $ 0.0207 \pm 0.0110  $  & $ 40.26 \pm 5.10  $  & $ 9.80 \pm 0.64  $  \\
& iEEA & $ 2.6564 \pm 1.1370  $  & $ 0.0171 \pm 0.0108  $  & $ 28.70 \pm 4.28  $  & $ 9.64 \pm 0.60  $  \\\hline\hline

\multicolumn{6}{c}{$\rho=0.9$}\\\hline
0.5 &  \multicolumn{5}{c}{$ \rhat =   1.54 \pm 0.5  $ }\\\hline
& DPP & $ 1.1819 \pm 0.1063  $  & $ 0.0094 \pm 0.0039  $  & $ 15.08 \pm 3.69  $  & $ 7.72 \pm 1.69  $  \\
& DPP.cv & $ 1.1779 \pm 0.1047  $  & $ 0.0090 \pm 0.0036  $  & $ 20.72 \pm 8.87  $  & $ 7.72 \pm 1.67  $  \\
& TSVD & $ 1.2951 \pm 0.1878  $  & $ 0.0119 \pm 0.0048  $  & $ 5.64 \pm 3.19  $  & $ 9.24 \pm 3.63  $  \\
& EEA & $ 1.1712 \pm 0.0985  $  & $ 0.0091 \pm 0.0037  $  & $ 12.72 \pm 5.74  $  & $ 8.82 \pm 2.47  $  \\
& iEEA & $ 1.1668 \pm 0.0969  $  & $ 0.0090 \pm 0.0038  $  & $ 9.98 \pm 4.06  $  & $ 8.40 \pm 1.90  $  \\\hline

1 &  \multicolumn{5}{c}{$ \rhat =  2 \pm 0  $ }\\\hline
& DPP & $ 1.4498 \pm 0.2353  $  & $ 0.0299 \pm 0.0130  $  & $ 21.88 \pm 3.42  $  & $ 9.54 \pm 0.79  $  \\
& DPP.cv & $ 1.4463 \pm 0.2509  $  & $ 0.0286 \pm 0.0137  $  & $ 27.34 \pm 7.21  $  & $ 9.54 \pm 0.79  $  \\
& TSVD & $ 2.1500 \pm 0.6380  $  & $ 0.0505 \pm 0.0249  $  & $ 7.84 \pm 2.78  $  & $ 12.60 \pm 4.73  $  \\
& EEA & $ 1.5153 \pm 0.3527  $  & $ 0.0306 \pm 0.0157  $  & $ 25.12 \pm 7.17  $  & $ 10.32 \pm 2.11  $  \\
& iEEA & $ 1.5937 \pm 0.4998  $  & $ 0.0342 \pm 0.0206  $  & $ 17.10 \pm 5.33  $  & $ 9.80 \pm 0.93  $  \\\hline\hline

\end{tabular}
\end{table}

\begin{table}[!tb]
\caption{\label{tab:simu2}
Performance of five methods: means and standard deviations of prediction errors, estimation errors and sizes of selected models across 50 replications.
Simulation setting 2: $n=100$, $m=50$, $p=25$, $s=|\supp(A)|=15$, $k=|\supp(A')|=25$, $r=5$.}
\smallskip
\centering
\begin{tabular}{cccccc}
\hline\hline
$b$ & Method 
& $\|\Yhat-Y\|^2_F/(mn)$ & $\|\Ahat-A\|^2_F/(mp)$ & $|\supp(\Ahat)|$ & $|\supp(\Ahat')|$ \\\hline\hline

\multicolumn{6}{c}{$\rho=0.1$}\\\hline
0.2 &  \multicolumn{5}{c}{$ \rhat =  4.74 \pm 0.44  $ }\\\hline
& DPP & $ 1.0759 \pm 0.0273  $  & $ 0.0030 \pm 0.0008  $  & $ 17.16 \pm 1.28  $  & $ 24.42 \pm 0.73  $  \\
& DPP.cv & $ 1.0605 \pm 0.0265  $  & $ 0.0023 \pm 0.0008  $  & $ 24.46 \pm 0.73  $  & $ 24.40 \pm 0.78  $  \\
& TSVD & $ 1.3859 \pm 0.1397  $  & $ 0.0157 \pm 0.0054  $  & $ 13.26 \pm 1.84  $  & $ 30.24 \pm 5.34  $  \\
& EEA & $ 1.0894 \pm 0.0271  $  & $ 0.0035 \pm 0.0007  $  & $ 15.36 \pm 0.63  $  & $ 27.34 \pm 1.88  $  \\
& iEEA & $ 1.0883 \pm 0.0268  $  & $ 0.0035 \pm 0.0007 $  & $ 15.18 \pm 0.48  $  & $ 25.92 \pm 1.31  $  \\\hline

0.4 &  \multicolumn{5}{c}{$ \rhat =   5 \pm 0  $ }\\\hline
& DPP & $ 1.0729 \pm 0.0245  $  & $ 0.0029 \pm 0.0004  $  & $ 17.58 \pm 1.46  $  & $ 24.98 \pm 0.14  $  \\
& DPP.cv & $ 1.0488 \pm 0.0232  $  & $ 0.0019 \pm 0.0002  $  & $ 24.66 \pm 0.59  $  & $ 24.98 \pm 0.14  $  \\
& TSVD & $ 1.2569 \pm 0.1423  $  & $ 0.0105 \pm 0.0056  $  & $ 15.20 \pm 0.78  $  & $ 29.90 \pm 5.04  $  \\
& EEA & $ 1.0734 \pm 0.0277  $  & $ 0.0029 \pm 0.0005  $  & $ 15.28 \pm 0.61  $  & $ 27.00 \pm 1.69  $  \\
& iEEA & $ 1.0733 \pm 0.0255  $  & $ 0.0030 \pm 0.0005  $  & $ 15.04 \pm 0.20  $  & $ 25.74 \pm 0.99  $  \\\hline\hline

\multicolumn{6}{c}{$\rho=0.9$}\\\hline
0.2 &  \multicolumn{5}{c}{$ \rhat =   3.16 \pm 0.55  $ }\\\hline
& DPP & $ 1.1037 \pm 0.0274  $  & $ 0.0286 \pm 0.0059  $  & $ 14.82 \pm 1.93  $  & $ 22.34 \pm 2.19  $  \\
& DPP.cv & $ 1.0756 \pm 0.027  $  & $ 0.0179 \pm 0.0048  $  & $ 22.60 \pm 2.05  $  & $ 22.38 \pm 2.19  $  \\
& TSVD & $ 1.2566 \pm 0.0937  $  & $ 0.0444 \pm 0.0121  $  & $ 8.28 \pm 2.84  $  & $ 29.18 \pm 6.02  $  \\
& EEA & $ 1.0962 \pm 0.0249  $  & $ 0.0239 \pm 0.0046  $  & $ 13.66 \pm 2.12  $  & $ 29.66 \pm 4.16  $  \\
& iEEA & $ 1.0944 \pm 0.0273  $  & $ 0.0246 \pm 0.0052  $  & $ 12.61 \pm 1.78  $  & $ 25.98 \pm 2.45  $  \\\hline

0.4 &  \multicolumn{5}{c}{$ \rhat =   4.56 \pm 0.5  $ }\\\hline
& DPP & $ 1.1601 \pm 0.0393  $  & $ 0.0617 \pm 0.0131  $  & $ 16.72 \pm 1.28  $  & $ 24.88 \pm 0.33  $  \\
& DPP.cv & $ 1.0626 \pm 0.0315  $  & $ 0.0167 \pm 0.0057  $  & $ 24.30 \pm 0.76  $  & $ 24.88 \pm 0.33  $  \\
& TSVD & $ 1.6125 \pm 0.2274  $  & $ 0.1330 \pm 0.0370  $  & $ 13.28 \pm 2.84  $  & $ 34.98 \pm 6.24  $  \\
& EEA & $ 1.1056 \pm 0.0293  $  & $ 0.0286 \pm 0.0074  $  & $ 15.64 \pm 1.19  $  & $ 35.98 \pm 5.28  $  \\
& iEEA & $ 1.1167 \pm 0.0321  $  & $ 0.0328 \pm 0.0081  $  & $ 14.88 \pm 0.63  $  & $ 28.66 \pm 2.80  $  \\\hline\hline

\end{tabular}
\end{table}

\prettyref{tab:simu1} and \prettyref{tab:simu2} report the means and the standard deviations of prediction errors, estimation errors and sizes of selected models based on $50$ replications in each setting. 
It is noticed that DPP.cv has the best performance for almost all cases considered, while DPP with the estimated universal penalty level tends to choose a smaller model with slightly larger estimation errors. 
In some settings, DPP.cv was able to reduce the estimation errors by up to $40\%$ when compared to TSVD, EEA and iEEA. 
Note that when comparing prediction errors, the quantity that makes most sense is the excessive error an estimator makes in addition to the oracle error that one would make even when the true coefficient matrix is given.
In the current setting, the (normalized) oracle error is $1$.
In terms of the excessive prediction error, it is observed that the prediction accuracy of DPP.cv outperformed the other methods by a similar percentage.

\paragraph{Comparison under different noise distributions}
We now compare the performance of these methods on non-Gaussian data. 
To this end, we consider three different noise distributions: $\sqrt{3/5}t_5$, $\sqrt{4/5}t_{10}$ and 3 Uniform (the sum of three uniform $[-1,1]$ random variables).
Here, $t_\nu$ stands for the $t$-distribution with $\nu$ degrees of freedom. 
We note that all three distributions have been normalized to have unit variance.
\prettyref{tab:simu3} reports the simulation results for the second setting with $\rho=0.1, b=0.2$ and for all three noise distributions.
It shows that our methods, esp.~DPP.cv, preserve competitive performance even for non-Gaussian data.
Moreover, when compared with the corresponding performance measures on Gaussian data (the first section in \prettyref{tab:simu2}), we see that all the estimators were relatively robust to the noise distributions, though their performance (with the exception of TSVD) did degrade as the tail of the noise distribution gets heavier.

\begin{table}[tb]
\caption{\label{tab:simu3}Performance of five methods on non-Gaussian data. Simulation setting 2: $n=100$, $m=50$, $p=25$, $s=|\supp(A)|=15$, $k=|\supp(A')|=25$, $r=5$, $\rho=0.1$ and $b=0.2$ (the same as the first section of \prettyref{tab:simu2}). }
\centering
\begin{tabular}{cccccc}
\hline\hline
Noise dist. & Method 
& $\|\Yhat-Y\|^2_F/(mn)$ & $\|\Ahat-A\|^2_F/(mp)$ & $|\supp(\Ahat)|$ & $|\supp(\Ahat')|$ \\\hline\hline
$\sqrt{3/5}t_5$ &  \multicolumn{5}{c}{$ \rhat =   4.66 \pm 0.48   $ }\\\hline
& DPP & $ 1.0786 \pm 0.0280  $  & $ 0.0031 \pm 0.0007  $  & $ 17.48 \pm 1.36  $  & $ 25.60 \pm 1.29  $  \\
& DPP.cv & $ 1.0637 \pm 0.0286  $  & $ 0.0024 \pm 0.0007  $  & $ 24.28 \pm 1.01  $  & $ 25.62 \pm 1.28  $  \\
& TSVD & $ 1.3733 \pm 0.1324  $  & $ 0.0152 \pm 0.0052  $  & $ 13.52 \pm 1.47  $  & $ 28.18 \pm 4.55  $  \\
& EEA & $ 1.0906 \pm 0.0262  $  & $ 0.0036 \pm 0.0007  $  & $ 15.24 \pm 0.59  $  & $ 27.62 \pm 1.83  $  \\
& iEEA & $ 1.0899 \pm 0.0264  $  & $ 0.0035 \pm 0.0007  $  & $ 15.10 \pm 0.46  $  & $ 26.24 \pm 1.29  $  \\
\hline\hline
$\sqrt{4/5}t_{10}$ &  \multicolumn{5}{c}{$\rhat =  4.78 \pm 0.42  $}\\\hline  
& DPP & $ 1.0758 \pm 0.0292  $  & $ 0.0029 \pm 0.0007  $  & $ 17.70 \pm 1.34  $  & $ 24.50 \pm 0.86  $  \\
& DPP.cv & $ 1.0589 \pm 0.0286  $  & $ 0.0022 \pm 0.0006  $  & $ 24.50 \pm 0.68  $  & $ 24.50 \pm 0.91  $  \\
& TSVD & $ 1.4164 \pm 0.1443  $  & $ 0.0166 \pm 0.0057  $  & $ 12.82 \pm 2.54  $  & $ 30.42 \pm 5.35  $  \\
& EEA & $ 1.0921 \pm 0.0291  $  & $ 0.0036 \pm 0.0006  $  & $ 15.24 \pm 0.62  $  & $ 27.70 \pm 1.90  $  \\
& iEEA & $ 1.0874 \pm 0.0292  $  & $ 0.0034 \pm 0.0006  $  & $ 15.08 \pm 0.40  $  & $ 26.14 \pm 1.23  $  \\
\hline\hline
3 Uniform &  \multicolumn{5}{c}{$\rhat = 4.72 \pm 0.45  $ }\\\hline  
& DPP & $ 1.0781 \pm 0.0299  $  & $ 0.0030 \pm 0.0007  $  & $ 17.50 \pm 1.39  $  & $ 24.14 \pm 0.81  $  \\
& DPP.cv & $ 1.0627 \pm 0.0315  $  & $ 0.0023 \pm 0.0007  $  & $ 24.42 \pm 0.93  $  & $ 24.12 \pm 0.85  $  \\
& TSVD & $ 1.4232 \pm 0.1292  $  & $ 0.0170 \pm 0.0051  $  & $ 12.78 \pm 2.26  $  & $ 28.64 \pm 4.06  $  \\
& EEA & $ 1.0931 \pm 0.0308  $  & $ 0.0036 \pm 0.0007  $  & $ 15.24 \pm 0.52  $  & $ 27.36 \pm 1.63  $  \\
& iEEA & $ 1.0898 \pm 0.0309  $  & $ 0.0035 \pm 0.0008  $  & $ 15.14 \pm 0.35  $  & $ 25.76 \pm 1.20  $  \\
\hline\hline
\end{tabular}
\end{table}


\subsection{\textit{In vivo} Calcium Imaging Data}
Calcium imaging has become an increasingly important tool in neuroscience to track the activity of neuronal populations by recording the dynamics of the time-varying fluorescence of the neurons \citep{akerboom2012optimization, chen2013ultrasensitive}. When a neuron fires an electrical action potential (spike), calcium will enter the cell and change its fluorescent properties by attaching to genetically encoded calcium indicators. By recording the movies of fluorescence activities, researchers hope to identify and demix the regions of interest (ROIs) as well as extract spike traces \citep{pnevmatikakis2014structured, haeffele2014structured}. 

Following the spatiotemporal model in \cite{pnevmatikakis2014structured}, suppose an $l_1\times l_2$ area (2d imaging plane of an original 3d volume) containing $K$ neurons (possibly overlapping) is monitored for $T$ time frames. Let $c_i=(c_i(1), \cdots, c_i(T))^\prime\in\mathbb{R}^{T}$ be the calcium activity and $\omega_i\in \mathbb{R}^m$ ($m=l_1\times l_2$) be the spatial footprint (stacked by the monitored area) of the $i\Th$ neuron. Then the fluorescence intensity observed at time $t$ can be modeled as 
\begin{equation}
y_t=\sum_{i=1}^K \omega_ic_i(t)+z_t, \ 1\leq t\leq T,
\nonumber
\end{equation}
where $z_t\stackrel{iid}{\sim} N(0, \sigma^2I_m)$ is the noise vector at time $t$. In matrix notations, 
\begin{equation}
Y=C\Omega+Z,
\nonumber
\end{equation}
where $Y=(y_1, \cdots, y_T)^\prime\in\mathbb{R}^{T\times m}, \Omega=(\omega_1, \cdots, \omega_K)^\prime\in\mathbb{R}^{K\times m}, C=(c_1, \cdots, c_K)\in\mathbb{R}^{T\times K}, Z=(z_1, \cdots, z_T)^\prime\in\mathbb{R}^{T\times m}$. Let $s_i=(s_i(1), \cdots, s_i(T))^\prime\in\mathbb{R}^T$ be the spike trace of the $i\Th$ neuron. Then the calcium activity can be characterized by a simple first order autoregressive model, 
\begin{equation} 
c_i(t)=\gamma c_i(t-1)+s_i(t), \,1\leq t\leq T, 
\nonumber
\end{equation}
or equivalently ($c_i(0)=0$ by convention),
\begin{equation}
S=GC,
\nonumber
\end{equation}
where $S=(s_1, \cdots, s_K)\in\mathbb{R}^{T\times K}$ and 
$$G=
 \begin{pmatrix}
  1& 0& \cdots & 0 \\
-\gamma  & 1& \ddots & \vdots\\
  \vdots  & \ddots  & \ddots &0 \\
0 & \cdots& -\gamma & 1
 \end{pmatrix}\in\mathbb{R}^{T\times T}.$$
In this way, 
\begin{align}
Y=G^{-1} S\Omega+Z=XA+Z 
\label{product}
\end{align}
where $A=S\Omega$ is the spatiotemporal convolution matrix and $X=G^{-1}$ is the known design matrix\footnote{Following \cite{vogelstein2010fast}, $\gamma$ is set at $\gamma=1-1/(\mbox{frame\, rate})$.}. 
The support of $\Omega$ is the location of the neurons and the support of $S$ represents the time frames when the neurons fire. Because the number of neurons in the monitored area is small and the neurons do not fire very frequently,  $\Omega$ is approximately row sparse and $S$ is approximately column sparse, which together imply that $A$ is two-way sparse (also low-rank by definition). 
Therefore, the generative model \eqref{product} can be viewed as a special case of model \eqref{eq:model} with $n=p=T$ and $m=l_1\times l_2$.  
To recover $\Omega$ and $S$, we suggest first estimating $A$ by the proposed algorithm and then running a nonnegative matrix factorization (NMF) on $\widehat{A}$ to obtain $\widehat{\Omega}$ and $\widehat{S}$. 
\citet{pnevmatikakis2014structured} proposed an alternating $l_1$ minimization strategy to estimate $\Omega$ and $S$ but no theoretical guarantee has been established for such heuristic. 
When signal-to-noise ratio is not high, their algorithm could be sensitive to initialization and could converge to a local minimum and yield suboptimal result. 
The DPP procedure proposed here is more robust because applying a denoising step in the first place removes most of the noise and hence the subsequent matrix factorization is less sensitive to initialization. 
\begin{figure}[tb]
  \centering
      \includegraphics[width=0.9\textwidth]{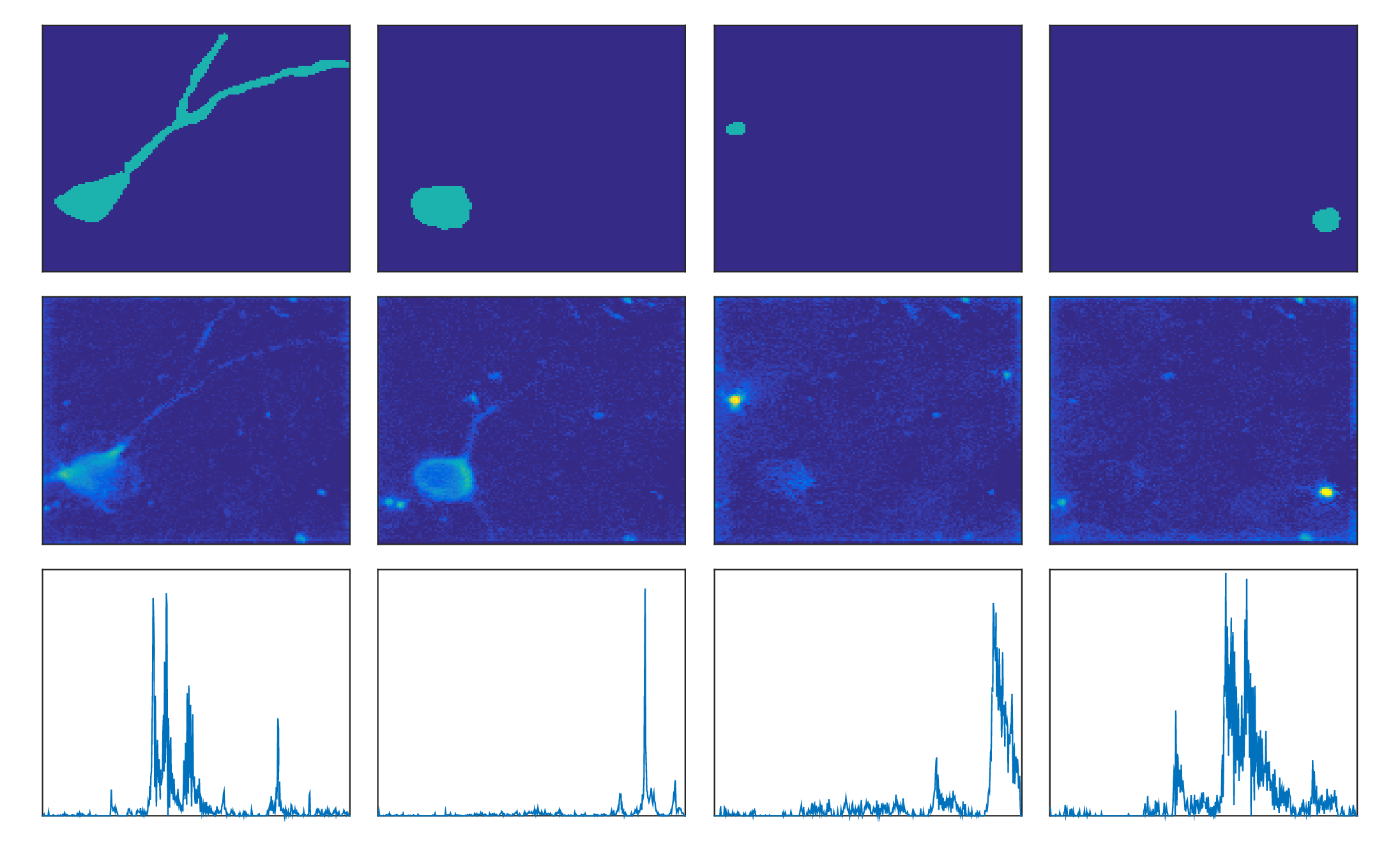}
  \caption{
Application to \textit{in vivo} calcium imaging data. Top: manually segmented regions of neurons. Middle: heat maps of the recovered spatial components. Bottom: estimated spike trace.}
  \label{neuron-results}
\end{figure}

The calcium imaging data ($n=p=T=559$, $m=135\times 131$) we use here is taken \textit{in vivo} from the primary auditory cortex of a mouse with genetically encoded calcium indicator GCaMP5 \citep{akerboom2012optimization}. 
We report here four most significant neurons to demonstrate the effectiveness of the proposed method as illustrated in Figure~\ref{neuron-results}. The top panel shows the manually segmented regions of the neurons from the raw dataset, which can be approximately regarded as the true support of the spatial component $\Omega$. The first neuron consists of a cell body with a dendritic branch and it heavily overlaps with the second neuron, making manual segmentation very challenging. The middle panel displays the heat maps of the recovered neurons by the proposed approach and they match the manual segmentation very well. The bottom panel of Figure~\ref{neuron-results} shows the estimated spike traces.

\section{Theoretical Properties}
\label{sec:theory}


In this section, we present theoretical results for a slight variant of the proposed estimation scheme when the noise matrix $Z$ in \eqref{eq:model} has i.i.d.~Gaussian entries. Their proofs are deferred to \prettyref{sec:proof}.

\subsection{Minimax Upper Bounds}
To facilitate the discussion, we put the estimation problem in a decision--theoretic framework.
We are interested in estimating the coefficient matrix $A$ in model \eqref{eq:model} where $A$ is both two-way sparse and of low rank, and $Z$ has i.i.d.~$N(0,\sigma^2)$ entries.
Thus, we assume that $A$ belongs to the following parameter space
\bel{eq:para-space} 
\Theta(s,k, r, d, \gamma) &=&  \Big\{A\in \reals^{p\times m}: \rank(A) = r,
\gamma d\geq \sigma_1(A)\ge\dots\ge \sigma_r(A) >d >0,  
\cr&&~~~~~~~~~~~~~~ \quad |\supp(A)|\leq s, |\supp(A')|\le k \Big\},
\eel
where $\supp(M)$ is the index set of nonzero rows in matrix $M$. 
Here and after, we treat $\gamma$ as an absolute positive constant.
To measure the accuracy of any estimator $\wt{A}$, we consider the following class of squared Schatten norm losses:
\begin{align}
\label{eq:loss}
L_q(A, \wt{A}) = \sqnorm{\wt{A} - A}^2, \qquad q\in [1,2].
\end{align}
For simplicity, we assume the noise variance $\sigma^2$ is known.
In addition, we treat the design matrix $X$ as fixed and the only source of randomness is the noise matrix $Z$.
In what follows, we present high probability error bounds for (a slight variant of) the DPP estimator where independent samples are generated and used in steps 1--4.
We believe the deviation from \prettyref{algo:est} is an artifact of the proof technique. 
Numerical studies (not reported) showed that the algorithm produces comparable results whether independent samples are used or a single sample is used repeatedly.

\paragraph{Independent sample generation}
Note that we can generate the desired independent samples from the observed $(X,Y)$ when the noises are homoscedastic and Gaussian.  
Indeed, when the entries of the noise matrix $Z$ are i.i.d. $N(0,\sigma^2)$, we can first generate an independent copy $\Ztil$ so that all entries in $Z+\Ztil$ and $Z-\Ztil$ are mutually independent and all follow the same Gaussian distribution $N(0, 2\sigma^2)$. 
Thus, $Y+\wt{Z}$ and $Y-\wt{Z}$ are independent, following model \eqref{eq:model} with i.i.d.~$N(0,2\sigma^2)$ noises. 
Employing this trick twice, we can generate four independent copies of responses
\bes
Y_{(i)}=XA+Z_{(i)}, \quad i=0,1,2,3,
\ees
where $Z_{(i)}$ has i.i.d. $N(0,\tsigma^2)$ entries with $\tsigma=2\sigma$. 
In the rest of this paper, when we mention \prettyref{algo:est}, we refer to the procedure with independent samples $Y_{(i)}$ used in the $(i+1)\Th$ step, $i=0,1,2,3$, where the noise variance is $\tsigma^2=4\sigma^2$.

\paragraph{The design matrix}
Without of loss of generality, we assume $X$ is of full rank.
Otherwise, we can always perform the following operation to reduce to the full rank case.
If $\rank(X) = q < n\wedge p$ and let $O\in \reals^{n\times q}$ be its left singular vector matrix. 
Setting $\wt{Y} = O'Y$ and $\wt{X} = O'X$, we obtain that $\wt{Y}$ and $\wt{X}$ satisfy model \eqref{eq:model} with the same coefficient matrix $A$, i.i.d.~$N(0,\sigma^2)$ noises and a design matrix of full rank.

We write the singular value decomposition of $XA$ as 
\begin{align}
	\label{eq:XA-svd}
XA = U\Delta V'
\end{align}
with $U\in O(n,r)$, $V\in O(m,r)$ and $\Delta = \diag(\delta_1,\dots, \delta_r)$ collects the non-zero singular values of $XA$. 
To introduce appropriate assumptions on $X$, we first make the following definition. 

\begin{definition}
\label{def:riesz}
For any $k\in [p]$,
the \emph{$\ell$-sparse Riesz constants} $\kappa_{\pm}(\ell)$ of $X$ are defined as 
\begin{align}
\label{eq:kappa}
\kappa_-^2(\ell;X) = \min_{B\subset [p], |B| = \ell} \sigma_{\min}(X_{*B}'X_{*B}),\qquad
\kappa_+^2(\ell;X) = \max_{B\subset [p], |B| = \ell} \sigma_{\max}(X_{*B}'X_{*B})
\end{align}
\end{definition}

By definition, if the $\ell$-sparse Riesz constants of $X$ are $\kappa_\pm(\ell;X)$, then for any $l\in [\ell]$, the $l$-sparse Riesz constants $\kappa_\pm(l;X)$ of $X$ satisfy $\kappa_-(\ell;X)\leq \kappa_-(l;X)\leq \kappa_+(l;X)\leq \kappa_+(\ell;X)$.

To establish upper bounds for the proposed estimator, for some integer $s_*$  depending only on $s$, we require the $s_*$-sparse Riesz constants of $X$ to satisfy the following condition. 

\begin{condition}[Sparse eigenvalue condition]
\label{cond-speigen}
There exist positive constants $s_*$ and $c_*$ and $K\geq 1$, such that the $s_*$-sparse Riesz constants satisfy
$K^{-1}\leq \kappa_-(s_*;X) \leq \kappa_+(s_*;X)\leq K$ and
\begin{align*}
\frac{\kappa^2_+(s_*;X)-\kappa^2_-(2s_*;X)}{\kappa^2_-(s_*;X)}<c_*.
\end{align*}
\end{condition}

We do not put condition on $\kappa_-(2s_*;X)$. Following the above definition and discussion, we know that $0\leq \kappa_-(2s_*;X) \leq \kappa_-(s_*;X)$ always holds.

\bigskip

The following theorem gives high probability upper bounds, provided that the design matrix satisfies mild regularity conditions and the penalty level is properly chosen. 

\begin{theorem}
\label{thm:upper}
Let  $A\in \Theta(s, k,r, d,\gamma)$ where $s\geq r\geq 1$.
Set the penalty level
\begin{align}
	\label{eq:lambda}
\lam = 4\sigma\max_{ j\le p}\|X_{*j}\|(\sqrt{r}+\sqrt{4\log (p\vee m)})
\end{align}
in steps 2 and 4 of \prettyref{algo:est} with the group Lasso penalty \eqref{eq:penlasso}.
Let $\alpha=2\sqrt{3}$ and $\beta=1.1$ in \prettyref{algo:est}.
Suppose that Condition \ref{cond-speigen} holds with an absolute constant $K>1$ for all $X$ and positive constants $s_*, c_*$ satisfying
\begin{equation}
	\label{eq:sparse-eig}
	s_*\ge 2s,\quad 
	6c_*\le \sqrt{s_*/s-1},
\end{equation}
and that there exist sufficiently small constants $c_0 > 0$ and $c_1 > 0$ such that 
\begin{align}
	\label{eq:init-cond}
\frac{2\sigma}{d}\Big\{ \sqrt{n}+\sqrt{k}+2\sqrt{\log (p\vee m)} + \sqrt{k\sqrt{n\log (p\vee m)}}\Big\} \le c_0, \quad
\sqrt{s}\lam/d\leq c_1.
\end{align}
Then uniformly over $\Theta(s,k,r,d,\gamma)$ in \eqref{eq:para-space}, with probability at least $1-3(p\vee m)^{-1}$, the output $\wh{A}$ of \prettyref{algo:est} satisfies
\bes
L_q(A, \wh{A})
\leq C \sigma^2 r^{2/q-1} (k+s)(r+\log (p\vee m)),
\qquad \text{for all $q\in [1, 2]$}
\ees
where $C$ is a constant depending only on $\kappa_\pm(s_*),  \gamma, c_*, c_0$ and $c_1$.
\end{theorem}



\subsection{Minimax Lower Bounds}

To assess the tightness of the error bounds in \prettyref{thm:upper},
we now provide minimax risk lower bounds for estimating $A$ under all loss functions in \eqref{eq:loss}.
\begin{theorem}
\label{thm:lower}
Let the observed $X,Y$ be generated by \eqref{eq:model} with $Z$ having i.i.d.~$N(0,\sigma^2)$ entries.
Suppose that the coefficient matrix $A\in \Theta(s,k,r,d,\gamma)$ for some $k\geq 2r$ and $s\geq 2r$ and that the $(2s)$-sparse Riesz constants of the design matrix $X$ satisfy 
$K^{-1}\leq \kappa_-(2s) \leq \kappa_+(2s) \leq K$ for some absolute constant $K>1$.
Then there exists a positive constant $c$ depending only on $\gamma$ and $\kappa_+(2s)$ such that the minimax risk for estimating $A$ satisfies
\begin{align}
	\label{eq:lower}
\inf_{\wh{A}} \sup_{\Theta}
\Expect L_q(A,\wh{A})
\geq c\sigma^2 \sth{ \pth{r^{2/q-1} \frac{d^2}{\sigma^2}} \wedge 
\qth{r^{2/q}(s+k) + r^{2/q-1} \pth{s\log \frac{\eexp p}{s} + k\log \frac{\eexp m}{k}}} },
\end{align}
for all $q\in [1,2]$.
\end{theorem}

\begin{remark}
Comparing \prettyref{thm:upper} and \prettyref{thm:lower}, we find that they match up to a multiplicative log factor in general and up to a constant multiplier when $r$ is no smaller than $\log(p\vee m)$ in order.
Therefore, under the conditions of \prettyref{thm:upper}, \prettyref{algo:est} attains nearly optimal convergence rates adaptively for all losses in \prettyref{eq:loss}. 

%
As we have mentioned earlier, the one-way sparse reduced rank regression model considered in the literature, such as \cite{Chen12}, \cite{Bunea12}, \cite{she2014selectable} and \cite{MaSun14}, 
does not consider column sparsity in $A$ and can be viewed as a special case of model \eqref{eq:model} with $k=m$.
In view of the foregoing discussion, our estimator is also adaptive to this special case while retaining the ability of fully exploiting potential column sparsity.
\end{remark}

\section{Conclusion and Discussion}
\label{sec:conclusion}

In this paper, we have proposed a new Double Projected Penalization (DPP) estimator for the coefficient matrix in two-way sparse reduced-rank regression.
The model is well motivated by massive datasets arising in a number of application fields, especially genomics and neurosciences.
The proposed estimator is fast to compute and demonstrates competitive performance when compared with existing methods in simulation studies.
In addition, we have illustrated its potential use in neuroscience by applying it to the analysis of a calcium imaging dataset.
Last but not least, we have further justified its nice empirical performance by a decision-theoretic analysis when the data is Gaussian.

In terms of the DPP estimator, an interesting problem to be studied in future is to establish high probability error bounds when the data is not Gaussian.
Since one cannot easily generate independent samples in such cases, we anticipate that different proof techniques will be needed to achieve this goal.
In addition, it is worth noting that steps 3--4 of \prettyref{algo:est} can be iterated till certain convergence criterion is met.
Thus, we could also define an iterative projected penalization estimator. 
However, based on simulation results not reported here, we did not find significant performance gain by employing such an iterative scheme, which is more costly in terms of computation.


Another potential direction for future research is to consider certain nonlinear extensions of the model. 
When the response is univariate, researchers have considered sparse sliced inverse regression \citep{Li12,Lin15}. 
It would be of great interest to conduct analogous investigations for multiple responses where both low-rankness and sparsity are involved.

\section{Proofs}
\label{sec:proof}


\subsection{Proof of \prettyref{thm:upper}}
\paragraph{Analysis of $V_{(0)}$.}
We first study the property of the right singular vector matrix $V^{(0)}$ obtained in the column-thresholding step of Stage I.
For $0<a_-<1<a_+$, define
\bes
J_{(0)}^\pm = \Big\{j: \|XA_{*j}\|^2\ge \tsigma^2 a_\mp\alpha \sqrt{n\log (p\vee m)}) \Big\}.
\ees
More specifically, let $a_-=0.1$ and $a_+=2$ in  the proof. Recall that  $\alpha=\sqrt{12}$ and $\tsigma=2\sigma$.


\begin{lemma}\label{lmm:sel0}[Stage I column selection] With probabbility at least $1-4(p\vee m)^{-2}$,
\bes
J_{(0)}^- \subset J_{(0)} \subset J_{(0)}^+
\ees
\end{lemma}
 
\begin{proof}[Proof of \prettyref{lmm:sel0}] Due to Gaussianity, $\|Y_{*j}^{(0)}\|^2/\tsigma^2$ follows a non-central $\chi^2$ distribution with $n$ degrees of freedom and noncentrality parameter $\|XA_{*j}\|^2/\tsigma^2$. By \prettyref{lmm:chisq},
\bes
P(J_{(0)}^- \not\subset J_{(0)}) 
&\le&  \sum_{j \in J_{(0)}^-} P\Big\{ \|Y_{*j}^{(0)}\|^2 < \tsigma^2(n+\alpha \sqrt{n\log (p\vee m)}) \Big\} \\
&\le& m P\Big\{\|Y_{*j}^{(0)}\|^2 < \tsigma^2 n+\|XA_{*j}\|^2 -  \tsigma^2 \alpha(a_+-1) \sqrt{n\log (p\vee m)}\ \Big|\ j \in J_{(0)}^- \Big\} \\
&\le& 2m \exp\Big(-\frac{\alpha^2(a_+-1)^2 n\log (p\vee m)}{4(\sqrt{n}+ (a_+\alpha)^{1/2} (n\log (p\vee m))^{1/4})^2}\Big)\\
&\le& 2(p\vee m)^{-2}.
\ees
Similarly, it is proved that $J_{(0)}\subset J_{(0)}^+$ holds with probability at least $1- 2(p\vee m)^{-2}$.

\end{proof}

\begin{lemma}\label{lmm:chisq} 
Let $X$ follow a non-central chi-square distribution $\chi^2_n(\lam)$ with $n$ degrees of freedom and non-centrality parameter $\lam$. Then
\bes
&&P\Big\{X\ge (n+\lam)+ 2(\sqrt{n}+\sqrt{\lam})s\Big\} \le \Big(1+\frac{1}{\sqrt{2}s}\Big)\exp(-s^2), \quad 
\text{if } 0\le s\le \frac{1}{2}n^{9/16},\\
&&P\Big\{X\le (n+\lam)- 2(\sqrt{n}+\sqrt{\lam})s\Big\} \le 2\exp(-s^2), \quad 
\text{if } 0\le s\le \frac{1}{2}n^{1/2}.
\ees
\end{lemma}

\begin{lemma}\label{lmm:opnorm}
Let $X$ be an $n\times m$ matrix with iid standard Gaussian entries. Then for any $t>0$,
\bes
P\Big\{ \|X\|>\sqrt{n}+\sqrt{m}+t \Big\} \le \exp(-t^2/2).
\ees
\end{lemma}

\begin{lemma}\label{lmm:V0}[Stage I subspace estimation]
With probability at least $1-3(p\vee m)^{-2}$,
\bes
\|VV'-V_{(0)}{V_{(0)}}'\| &\le& \frac{C_1\tsigma}{d}\Big\{ \sqrt{n}+\sqrt{k}+2\sqrt{\log (p\vee m)} + \sqrt{k\sqrt{n\log (p\vee m)}}\Big\}, \\
\|VV'-V_{(0)}{V_{(0)}}'\|_F &\le& \frac{C_2\tsigma}{d}\Big\{ \sqrt{r}(\sqrt{n}+\sqrt{k}+2\sqrt{\log (p\vee m)}) + \sqrt{k\sqrt{n\log (p\vee m)}}\Big\}.
\ees
\end{lemma}

\begin{proof}[Proof of \prettyref{lmm:V0}] 
We study the upper bounds in the event where $J_{(0)}^- \subset J_{(0)} \subset J_{(0)}^+$ holds.
We may reorder the columns of matrices such that $XA - \Ytil^{(0)}$ is of the following form
\bes
XA - \Ytil^{(0)} = \begin{pmatrix}
-Z_{* J_{(0)} } &  UDV_{* J_{(0)}^c }'
\end{pmatrix}
\ees
\prettyref{lmm:opnorm} provides an upper bound for $\|Z_{* J_{(0)} }\|$ as follows
\bes
\|Z_{* J_{(0)} }\| \le \tsigma(\sqrt{n}+\sqrt{J_{(0)}}+2\sqrt{\log (p\vee m)}) \le \tsigma( \sqrt{n}+\sqrt{k}+2\sqrt{\log (p\vee m)})
\ees
with probability at least $1-(p\vee m)^2$, since $|J_{(0)}|\le |J_{(0)}^+| =k$. Moreover, it holds that, in the event of $J_{(0)}^- \subset J_{(0)}$,
\bes
\|U\Delta V_{* J_{(0)}^c }'\|^2 \le \|\Delta V_{* (J_{(0)}^-)^c }'\|_F^2 \le   \tsigma^2 a_-\alpha  k \sqrt{n\log (p\vee m)}.
\ees
Thus, we have
\bes
\|XA - \Ytil^{(0)}\|\le \tsigma( \sqrt{n}+\sqrt{k}+2\sqrt{\log (p\vee m)})+\tsigma \sqrt{a_-\alpha  k \sqrt{n\log (p\vee m)}}
\ees
and the desired results then follows from the $\sin\theta$ theorem.
\end{proof}

\paragraph{Analysis of $U_{(1)}$.}

\begin{lemma}\label{lmm:U1}[Stage I Regression]
Under the condition of \prettyref{thm:upper}, there exists a constant $C$ depending only on $\kappa_\pm(s_*), c_*$ and $c_0$, such that with probability at least $1-(p\vee m)^{-1}$,
\begin{align*}
\|U_{(1)}U_{(1)}'-UU'\|_F \leq C \sqrt{s}\lam/d.
\end{align*}
\end{lemma}

\begin{proof}[Proof of \prettyref{lmm:U1}] 
Let $U_*\in \reals^{n\times r}$ be the left singular vector matrix of $XAV_{(0)}=UDV'V_{(0)}$. 
Under condition \eqref{eq:init-cond}, $V'V_{(0)}$ is an $r\times r$ matrix of full rank, 
and so the column space of $U^*$ is the same as the column space of $U$; \ie, $U_*U_*'=UU'$.
By Wedin's $\sin \theta$ Theorem \citep{Wedin72}, 
\bes
\|U_{(1)}U_{(1)}'-UU'\|_F=\|U_{(1)}U_{(1)}'-U_*U_*'\|_F\le \frac{\|XB_{(1)}-XAV_{(0)}\|_F}{\sigma_r(XAV_{(0)})},
\ees
where $\sigma_r(XAV_{(0)})$ is the $r\Th$ singular value of $XAV_{(0)}$.

Since for any unit vector $x$,
\bes
 \|V'V_{(0)}x\|^2&=& x'V_{(0)}'VV'V_{(0)}x\\
&=& 1-x'V_{(0)}'(VV'-V_{(0)}V_{(0)}')V_{(0)}x\\
&\ge& 1- \|VV'-V_{(0)}V_{(0)}'\|.
\ees
Thus, we have
$\sigma^2_r(V'V_{(0)})=\min_{\|x\|=1} \|V'V_{(0)}x\|^2 \ge  1- \|VV'-V_{(0)}V_{(0)}'\|$. When $c_0$ is small enough, $\|VV'-V_{(0)}V_{(0)}'\|$ is sufficiently small by \prettyref{lmm:V0}. So there exists a constant $c'$ such that $ \sigma_r(V'V_{(0)}) >c'$.
Note that $XAV_{(0)} = XAVV' V_{(0)}$, and so
\begin{align*}
\sigma_r(XAV_{(0)}) \geq \sigma_r(XAV) \sigma_r(V'V_{(0)}) \geq  \delta_r c',
\end{align*}
where the last inequality holds under condition \eqref{eq:init-cond} since $\sigma_r(XAV) = \sigma_r(XA) = \delta_r$.
Further note that
\begin{align*}
\|XB_{(1)}-XAV_{(0)}\|_F \leq \kappa_+(2s)\| B_{(1)} - AV_{(0)}\|_F
\leq \kappa_+(s_*)\| B_{(1)} - AV_{(0)}\|_F
\end{align*}
and that $\delta_r \geq \kappa_-(s) \sigma_r(A)\geq \kappa_-(s_*) d$,
the desired result then follows from Part (ii) of \prettyref{thm:glasso} with $\eta=1/(p\vee m)$.
\end{proof}

\paragraph{Analysis of $V_{(1)}$.} 

Recall
\bes
J_{(1)} = J_{(0)} \cup \Big\{j: \|{U_{(1)}}'Y_{*j}^{(2)}\|^2\ge \beta \tsigma^2(r+ 2\sqrt {3r\log (p\vee m)}+6\log (p\vee m)) \Big\}.
\ees
For $b_-<b_+$, define
\bes
J_{(1)}^\pm = \Big\{j: \|XA_{*j}\|^2\ge \tsigma^2 b_\mp(r+2 \sqrt{3r\log (p\vee m)} +6\log (p\vee m)) \Big\}.
\ees
More specifically, let $b_+=4.5$ and $b_-=0.002$ in the proof. Recall that  $\beta=1.1$.


\begin{lemma}\label{lmm:chisq2}
Let $X$ follow a chi-square distribution $\chi^2_n$ with $n$ degrees of freedom. Then for any $t>0$
\bes
P(X> n+2\sqrt{n}t+2t^2)< \exp(-t^2)
\ees
\end{lemma}

\begin{lemma}\label{lmm:sel1}[Stage II column selection]
Assume $\|U_{(1)} U_{(1)}'-UU'\|<c$ for some small positive constant $c<0.05$.
With probabbility at least $1-2(p\vee m)^{-2}$,
\bes
J_{(1)}^- \subset J_{(1)} \subset J_{(1)}^+
\ees
\end{lemma}

\begin{proof}[Proof of \prettyref{lmm:sel1}]
For $j\in J_{(1)}^-\setminus J_{(0)}$,
\bes
\| U_{(1)}'Y_{*j}^{(2)}\| &=&  \| U_{(1)}'(UDV_{*j}' + Z_{*j}^{(2)})\| \cr
&\ge& \| U_{(1)}'UDV_{*j}'\| -  \| U_{(1)}'Z_{*j}^{(2)}\|
\ees 
The first term is
\bes
 \| U_{(1)}'UDV_{*j}'\|^2
 &\ge& \|XA_{*j}\|^2(1-\|U_{(1)} U_{(1)}'-UU'\|)\ge \|XA_{*j}\|^2(1-c)\\
&\ge&  \tsigma^2 (1-c)b_+(r+2 \sqrt{3r\log (p\vee m)} +6\log (p\vee m)) 
\ees
Since $ U_{(1)}'Z_{*j}^{(2)}\sim N(0,\tsigma^2I_r)$, it follows from \prettyref{lmm:chisq2} that
\bes
\| U_{(1)}'Z_{*j}^{(2)}\|^2 \le \tsigma^2(r+ 2\sqrt {3r\log  (p\vee m)}+6\log  (p\vee m)),
\ees
with probability at least $1- (p\vee m)^{-3}$. Thus, in the same event, we have
\bes
\| U_{(1)}'Y_{*j}^{(2)}\|
&\ge& (\sqrt{(1-c)b_+}-1)\tsigma\Big\{r+ 2\sqrt {3r\log(p\vee m)}+6\log(p\vee m)\Big\}^{1/2}\\
&\ge& \beta^{1/2}\tsigma(r+ 2\sqrt {3r\log(p\vee m)}+6\log(p\vee m))^{1/2},
\ees
due to $(\sqrt{(1-c)b_+}-1)^2>(\sqrt{0.95\times 4.5}-1)^2>1.1=\beta$.
Hence, we have $j\in J_{(1)}$. So it holds that $ J_{(1)}^-\subset  J_{(1)}$ with probability at least
$1-(p\vee m)^{-2}$. Similarly, we have $ J_{(1)}\subset  J_{(1)}^+$ with probability at least $1-(p\vee m)^{-2}$,  
due to $(\sqrt{(1+c)b_-}+1)^2<1.1=\beta$.
\end{proof}

\begin{lemma}\label{lmm:V1}[Stage II subspace estimation]
Suppose $\|U_{(1)}U_{(1)}'-UU'\|_F<c_1'$ for a sufficiently small positive constant $c_1'$.
Then there exists a constant $C$ depending only on $\kappa_\pm(s_*),  \gamma$ and $c_1'$ such that with probability at least $1-(p\vee m)^{-1}$,
\bes
\|V_{(1)} V_{(1)}'-VV'\|_F \le C\sigma\sqrt{(k+s)(r+\log (p\vee m)) }/d
\ees
\end{lemma}

\begin{proof}[Proof of \prettyref{lmm:V1}]
\bel{eq:upperV1-1}
\|V_{(1)} V_{(1)}'-VV'\|_F \le \frac{\|U_{(1)} U_{(1)}'\Ytil^{(1)}-U_{(1)} U_{(1)}'XA \|_F}{\sigma_r(U_{(1)} U_{(1)}'XA)}.
\eel
We first upper bound the numerator
\bel{eq:upperV1-2}
&&\|U_{(1)} U_{(1)}'\Ytil^{(1)}-U_{(1)} U_{(1)}'XA \|_F  \cr
&\le& \|U_{(1)}'(\Ytil^{(1)}_{*J_{(1)}}-XA_{*J_{(1)}}) \|_F + \|U_{(1)}U_{(1)}'XA _{*J_{(1)}^c} \|_F\cr
&\le& \|U_{(1)}'(\Ytil^{(1)}_{*J^{(1)}}-XA_{*J_{(1)}}) \|_F + \|(U_{(1)}U_{(1)}'-UU')XA _{*J_{(1)}^c} \|_F + \| UU' XA_{*(J_{(1)}^-)^c}\|_F \cr
&\le& \tsigma(\sqrt{rk}+\sqrt{\log (p\vee m)}) + d \|(U_{(1)}U_{(1)}'-UU'\| +  \tsigma\sqrt{k}\sqrt{b_+(r+2 \sqrt{3r\log (p\vee m)} +6\log (p\vee m))}\cr
&\le& C\sigma\sqrt{(k+s)(r+\log (p\vee m)) }
\eel
To lower bound the denominator, we apply Weyl's theorem to obtain
\begin{align*}
\sigma_r(U_{(1)}U_{(1)}'XA) 
& \geq \sigma_r(UU'XA) - \opnorm{U_{(1)}U_{(1)}'XA - UU'XA} \\
& \geq \delta_r - \opnorm{U_{(1)}U_{(1)}' -UU'}\opnorm{XA}.
\end{align*}
Note that $\delta_r \geq \kappa_-(s_*)d$, $\opnorm{XA}\leq \kappa_+(s_*)\gamma d$ and that $\opnorm{U_{(1)}U_{(1)}' -UU'}\leq \| U_{(1)}U_{(1)}' -UU'\|_F\leq c_1'$.
Thus, for sufficiently small value of $c_1'$, we obtain
\begin{align}
		\label{eq:upperV1-3}
\sigma_r(U_{(1)}U_{(1)}'XA) \geq C^{-1} d,
\end{align}
where $C>0$ is a constant depending only on $\kappa_\pm(s_*),\gamma$ and $c_1'$.
Combining \eqref{eq:upperV1-1} -- \eqref{eq:upperV1-3}, we complete the proof. 
\end{proof}

\paragraph{Proof of \prettyref{thm:upper}.}
\begin{proof} By the definition of $\wh{A}$, we have
\bes
\|\wh{A} - A \|_F &=&\|B_{(2)} V_{(1)}' - AVV' \|_F\cr
&\le& \|B_{(2)} V_{(1)}' - A V_{(1)} V_{(1)}' \|_F + \| A V_{(1)} V_{(1)}' - AVV' \|_F \cr
&\le& \opnorm{V_{(1)}}\| B_{(2)} - AV_{(1)} \|_F + \opnorm{A} \| V_{(1)} V_{(1)}' - VV' \|_F.
\ees
Assembling the bounds in all lemmas, 
\bel{pf-upper-fro}
\|\wh{A} - A \|_F^2 &\lesssim& \sigma^2(k+s)(r+\log (p\vee m))
\eel
The desired upper bound on other Schatten norm losses is a consequence of (\ref{pf-upper-fro}) and the inequality $\sqnorm{\wh{A}-A}^2 \leq (2r)^{2/q-1} \|\wh{A}-A\|_F^2$ for all $q\in [1,2]$.

\end{proof}

\subsection{Proof of \prettyref{thm:lower}}
\label{sec:proof-lower}

\newcommand{\dKL}{d_{\rm KL}}
\newcommand{\vol}{\mathrm{vol}}

\newcommand{\rmS}{{\rm S}}

For any probability distributions $P$ and $Q$, let $D(P||Q)$ denote the Kullback--Leibler divergence of $Q$ from $P$.
For any subset $K$ of $\reals^{m\times n}$, the volume of $K$ is 
$\vol(K) = \int_K \diff \mu$ where $\diff\mu$ is the usual Lebesgue measure on $\reals^{m\times n}$ by taking the product measure of the Lebesgue measures of individual entries.
With these definitions,
we state the following variant of Fano's lemma \citep{IKbook,Birge83,Tsybakov09}.
This version has been established as Proposition 1 in \cite{MaWu13}.
It will be used repeatedly in the proof of the lower bounds.
Throughout the proof, we denote $\kappa_+(2s)$ by $\kappa_+$.

\begin{proposition}
\label{prop:fano}
	Let $(\Theta, \rho)$ be a metric space and $\{P_\theta: \theta \in \Theta\}$ a collection of probability measures. 	For any totally bounded $T \subset \Theta$, denote by $\calM(T, \rho,\epsilon)$ the $\epsilon$-packing number of $T$ with respect to $\rho$, \ie, the maximal number of points in $T$ whose pairwise minimum distance in $\rho$ is at least $\epsilon$.
Define the \emph{Kullback-Leibler diameter} of $T$ by
	\begin{equation}
	\dKL(T) \triangleq \sup_{\theta,\theta' \in T} \KL{P_\theta}{P_{\theta'}}.
	\label{eq:dKL}
\end{equation}
	 Then
	\begin{equation}
	\inf_{\hat{\theta}} \sup_{\theta \in \Theta} \Expect_{\theta}[\rho^2(\hat{\theta}(X),\theta)] \geq \sup_{T \subset \Theta} \sup_{\epsilon > 0} \frac{\epsilon^2}{4} \pth{1 -  
	\frac{\dKL(T) + \log 2}{\log \calM(T,\rho,\epsilon)}}.
	\label{eq:fano}
\end{equation}
In particular, if $\Theta \subset \reals^d$ and 
$\norm{\cdot}$ is some norm on $\reals^d$, then 
\begin{equation}
	\inf_{\hat{\theta}} \sup_{\theta \in \Theta} \Expect_{\theta}[\|\hat{\theta}(X) - \theta\|^2] \geq \sup_{T \subset \Theta} \sup_{\epsilon > 0} \frac{\epsilon^2}{4} \pth{1 -  
	\frac{\dKL(T) + \log 2}{\log \frac{\vol(T)}{\vol(B_{\|\cdot\|}(\epsilon))}}}.
	\label{eq:fano-vol}
\end{equation}
\end{proposition}

We first prove an oracle version of the lower bound. 
One can think of it as an lower bound for the minimax risk when we know that the nonzero entries of the coefficient matrix $A\in \reals^{p \times m}$ are restricted to the top--left $s\times r$ block (or the top left $r\times k$ block). 
\begin{lemma}
\label{lmm:oracle-lowbd}
Let 
$\Theta_0(s,r,r,d,\gamma) \subset \Theta(s,k,r,d,\gamma)$
be the sub-collection of all matrices whose nonzero entries are in the top left $s\times r$ block.
Suppose $\sigma = 1$. 
There exists a positive constant $c$ that depends only on $\kappa_+$ and $\gamma$, such that
for any $q\in [1,2]$, the minimax risk for estimating $A$ over $\Theta_0$ satisfies
\begin{align*}
\inf_{\wh{A}} \sup_{\Theta_0} \Expect L_q(A,\wh{A})
\geq c \qth{(r^{2/q-1} d^2) \wedge  (r^{2/q} s)}.
\end{align*}

Similarly, let $\Theta_0'(r,k,r,d,\gamma) \subset \Theta(s,k,r,d,\gamma)$
be the sub-collection of all matrices whose nonzero entries are in the top left $r\times k$ block. Under the same conditions, we have 
\begin{align*}
\inf_{\wh{A}} \sup_{\Theta_0'} \Expect L_q(A,\wh{A})
\geq c \qth{(r^{2/q-1} d^2) \wedge  (r^{2/q} k)}.
\end{align*}
\end{lemma}
\begin{proof}
In what follows, we focus on proving the first claim and the second claim follows from essentially the same argument.

By a simple sufficiency argument, we can reduce to model \eqref{eq:model} with $p = s$ and $m = r$, which we assume in the rest of this proof without loss of generality.

Let $A_0 = \diag(1,\dots,1)\in \reals^{s\times r}$.
Moreover, for any $\delta$ and any $q\in [1,2]$, let
$B_{{\rm S}_q}(\delta) = \{{A} \in \reals^{s\times r}: \sqnorm{{A}} \leq \delta\}$ denote the Schatten-$q$ ball with radius $\delta$ in $\reals^{s\times r}$.
For some constant $a>0$ to be specified later, define
\begin{align}
	\label{eq:oracle-lowbd-space}
T(a) = \frac{\gamma d}{2} A_0 + B_{{\rm S}_2}(\sqrt{a})
= \sth{\frac{\gamma d}{2}A_0 + M: M \in B_{{\rm S}_2}(\sqrt{a})}.
\end{align}
For any $A_1, A_2\in T(a)$, we have
\begin{align*}
D(P_{A_1} || P_{A_2}) & = \frac{1}{2}\norm{X A_1 - X A_2}_{\rmS_2}^2 
\leq \frac{1}{2}\opnorm{X}^2 \norm{A_1-A_2}_{\rmS_2}^2
\leq 2\kappa_+^2 a.
\end{align*}
Here, the last inequality holds since $\opnorm{X}\leq \kappa_+$ under the assumption that $X\in \reals^{s\times r}$ and $\norm{A_1-A_2}_{\rmS_2}^2\leq 4a$ by definition \eqref{eq:oracle-lowbd-space}.
So
\begin{align}
	\label{eq:oracle-lowbd-dKL}
\dKL(T(a)) \leq  2\kappa_+^2 a.
\end{align}
By the inverse Santalo's inequality (see, \eg, Lemma 3 of \cite{MaWu13}), for some universal constants $c_0$,
\begin{align}
\vol(T(a))^{1\over sr} & = \vol(B_{\rmS_2}(\sqrt{a}))^{1\over sr} 
= \sqrt{a}\cdot \vol(B_{\rmS_2}(1))^{1\over sr} 
\nonumber \\
& \geq \sqrt{a}\cdot \frac{c_0}{\Expect \norm{{Z}}_{\rmS_2}}
\label{eq:vol-lowbd-1}\\
& \geq \sqrt{a}\cdot \frac{c_0'}{\sqrt{sr}}.
\label{eq:vol-lowbd-2}
\end{align}
In \eqref{eq:vol-lowbd-1}, ${Z}$ is a $s\times r$ matrix with i.i.d.~$N(0,1)$ entries.
The inequality in \eqref{eq:vol-lowbd-2} holds since by Jensen's inequality, $\Expect\norm{{Z}}_{\rmS_2} \leq \sqrt{\Expect\norm{{Z}}_{\rmS_2}^2} = \sqrt{sr}$.

On the other hand, by Urysohn's inequality (see, \eg, Eq.(19) of \cite{MaWu13}), for any $\epsilon > 0$ and $q\in [1,2]$,
\begin{align*}
\vol(B_{\rmS_q}(\epsilon))^{1\over sr} & \leq \frac{\epsilon \Expect \norm{{Z}}_{\rmS_{q'}}}{\sqrt{sr}} 
\leq \frac{\epsilon r^{\frac{1}{q'}} \Expect \opnorm{{Z}}}{\sqrt{sr}}
\leq 2\epsilon r^{\frac{1}{2} - \frac{1}{q}}.
\end{align*}
Here, $\frac{1}{q'} + \frac{1}{q} = 1$ and ${Z}$ is a $s\times r$ matrix with i.i.d.~$N(0,1)$ entries.
The last inequality is due to Gordon's inequality (see, \eg, \cite{Davidson01}): $\Expect\opnorm{{Z}}\leq \sqrt{s}+\sqrt{r}\leq 2\sqrt{s}$.

Now let 
\begin{align}
	\label{eq:oracle-lowbd-eps}
a = \pth{\frac{\gamma\wedge 2-1}{2}}^2 \pth{sr \wedge d^2},\quad \mbox{and}
\quad 
\epsilon = \frac{c_0'}{2\kappa_+} \sqrt{a}\, r^{\frac{1}{q}-\frac{1}{2}}.
\end{align}
Then for any ${A}\in T(a)$ and any $i\in [r]$, 
$|\sigma_i({A}) - \frac{\gamma}{2}d|\leq \sqrt{a} \leq \frac{\gamma\wedge 2 - 1}{2}d$, and so $\sigma_i({A})\in [d, \gamma d]$ and $T(a)\subset \Theta_0(s,r,d,\gamma)$.
Applying \prettyref{prop:fano} with $T(a)$ and $\epsilon$ in \eqref{eq:oracle-lowbd-space} and \eqref{eq:oracle-lowbd-eps}, we obtain a lower bound on the order of $\epsilon^2$.
This completes the proof.
\end{proof}


\begin{lemma}
\label{lmm:scatter}
Let $s\geq r$ be positive integers.
There exist a matrix ${W}\in \reals^{s\times r}$ and two absolute constants $c_0\in (\frac{1}{2},1)$ and $c_1 > 0$ 
such that $\fnorm{{W}}\leq 1$ and for any subset $B\subset [s]$ such that $|B|\geq c_0 s$, $\sqnorm{{W}_{B*}}\geq c_1 r^{\frac{1}{q}-\frac{1}{2}}$ for any $q\in [1,2]$.
\end{lemma}
\begin{proof}
We divide the proof into two cases, namely when $s\geq 25$ and when $s<25$.

$1^\circ$ When $s\geq 25$,
let ${Z} \in \reals^{s\times r}$ have i.i.d.~$N(0,1)$ entries.
Then $\fnorm{{Z}}^2\sim \chi^2_{sr}$, and \citet[Eq.(4.3)]{Laurent00}
implies that 
\begin{align*}
\Prob\sth{\fnorm{{Z}}^2 \geq sr + 2s\sqrt{r} + 2s} \leq \eexp^{-s}.
\end{align*}
Moreover, for any $c_0 > \frac{1}{2}$,
\begin{align*}
& \Prob\sth{\exists B\subset [s],\, \mbox{s.t.}\, |B| = c_0s\,\,\,\mbox{and}\,\,\, \sigma_r({Z}_{B*}) < \sqrt{c_0 s} - \sqrt{r} - \frac{1}{2}\sqrt{c_0 s}}\\
& \leq \sum_{B\subset [s], |B| = c_0 s}
\Prob\sth{\sigma_r({Z}_{B*}) < \sqrt{c_0 s} - \sqrt{r} - \frac{1}{2}\sqrt{c_0 s}}\\
& \leq {s \choose (1-c_0) s} \eexp^{-c_0 s/4}\\
& \leq \exp\sth{-s\qth{\frac{c_0}{4} + (1-c_0)\log(1-c_0)}}.
\end{align*}
Here, the first inequality is due to the union bound, the second inequality is due to the Davidson-Szarek bound, and the last inequality holds since for any $\alpha \in (\frac{1}{2},1)$, ${s\choose \alpha s}={s\choose (1-\alpha) s}\leq ({\eexp \over 1-\alpha})^{(1-\alpha) s}$.
If we set $c_0 \geq 0.96$, then the multiplier $\frac{c_0}{4} + (1-c_0)\log(1-c_0) \geq 0.1$. 

So when $c_0 = 0.96$ and $s\geq 25$, 
the sum of the right hand sides of the last two displays is less than $1$.
Thus, there exists a deterministic matrix ${Z}_0$ on which both events happen.
Now define ${W} = {Z}_0 / \fnorm{{Z}_0}$. 
Then $\fnorm{{W}} = 1$ by definition, and for any $B\subset [s]$ with $|B| = c_0 s$, 
\begin{align*}
\sqnorm{{W}_{B*}} & \geq r^{1/q} \sigma_r({W}_{B*}) \\
& = r^{1/q} \sigma_r(({Z}_0)_{B*}) / \fnorm{{Z}_0}\\
& \geq r^{1/q} \frac{\frac{1}{2}\sqrt{c_0 s} - \sqrt{r}}{\sqrt{sr + 2s\sqrt{r}+2r}}\\
& \geq c_1 r^{1/q-1/2}.
\end{align*}
Note that the last inequality holds with an absolute constant $c_1$ when $r \leq \frac{1}{8}c_0 s$.
When $r > \frac{1}{8}c_0 s$, we can always let $\tilde{r} = \frac{1}{8}c_0 r \leq \frac{1}{8}c_0 s$ and repeat the above arguments on the $s\times \tilde{r}$ submatrix of ${Z}$ consisting of its first $\tilde{r}$ columns, and the conclusion continues to hold with a modified absolute constant $c_1$.
This completes the proof for all subsets $B$ with $|B| = c_0 s$. 
The claim continues to hold for all $|B|\geq c_0 s$ since the Schatten-$q$ norm of a submatrix is always no smaller than the the whole matrix.

$2^\circ$
When $s< 25$, we have $r< 25$ since $r\leq s$ always holds. Let ${W} = \begin{bmatrix}
	\frac{1}{\sqrt{s}}\mathbf{1}_s & \bszero
\end{bmatrix}\in \reals^{s\times r}$, \ie, the first column of ${W}$ consists of $s$ entries all equal to $1/\sqrt{s}$ and the rest are all zeros. 
So ${W}$ is rank one.
It is straightforward to verify the desired conclusion holds since for any $B\subset [s]$, $\sqnorm{{W}_{B *}} = \fnorm{{W}_{B *}} = \sqrt{|B|/s}$.
This completes the proof.
\end{proof}

\begin{lemma}
\label{lmm:packing}
Let $a = d^2 \wedge s\log\frac{\eexp p}{s}$.
There exist three positive constants $c_1, c_2, c_3$ that depend only on $\gamma$ and $\kappa_+$, 
and a subset $\Theta_1 \subset \Theta(s,k,r,d,\gamma)$, such that
$c_3\leq c_2/3$, $\dKL(\Theta_1)\leq c_3 a$
and that for any $q\in [1,2]$,
\begin{align*}
\log \calM(\Theta_1, \sqnorm{\cdot}, c_1 \sqrt{a}\, r^{1/q-1/2} ) \geq c_2 s\log\frac{\eexp p}{s},
\end{align*}
where $\dKL$ is the Kullback--Leibler diameter and $\calM$ is the packing number defined in \prettyref{prop:fano}.

Similarly, for $b=d^2\wedge k\log\frac{\eexp m}{k}$, there is another subset $\Theta'\subset \Theta(s,k,r,d,\gamma)$ such that $\dKL(\Theta_1')\leq c_3 b$ and that for any $q\in [1,2]$,
\begin{align*}
\log \calM(\Theta_1', \sqnorm{\cdot}, c_1 \sqrt{b}\, r^{1/q-1/2} ) \geq c_2 k\log\frac{\eexp m}{k}.
\end{align*}
\end{lemma}
\begin{proof}
Let us focus on the first claim and we shall remark on how to establish the second claim at the end of this proof.

Let ${W}\in \reals^{(s-r)\times r}$ satisfy the conclusion of \prettyref{lmm:scatter} and define $s_0 = (1-c_0) (s-r)$.
Let $\calB = \sth{B_1,\dots, B_N}$ be a maximal set consisting of subsets of $[p]\backslash [r]$ with cardinality $s-r$ and for any $B_i\neq B_j$, $|B_i\cap B_j|\leq s_0$.
By Lemma A.3 of \cite{Rigollet11} and Lemma 2.9 of \cite{Tsybakov09}, there exists an absolute positive constant $c_2'$
such that
\begin{align*}
\log{N} \geq c_2' (s-r)\log\frac{\eexp (p-r)}{s-r}.
\end{align*}
Now for each $B_i\in \calB$, define ${W}^{(i)}\in \reals^{m\times n} $ by setting the submatrix ${W}^{(i)}_{B_i [r]} = {W}$ and filling the remaining entries with zeros.
Then for any $i\neq j$, $|B_i\cap B_j|\leq s_0$, and so there exists a set $B_{ij}\subset [s]$ with $|B_{ij}|\geq s-r-s_0 = c_0 (s-r)$, such that 
\[
\sqnorm{{W}^{(i)} - {W}^{(j)}} \geq \sqnorm{{W}_{B_{ij} *}} \geq c_1' r^{1/q-1/2},
\]
where $c_1'$ is an absolute constant due to \prettyref{lmm:scatter}.

Define $M_0 = \begin{bmatrix}
I_r & 0 \\ 0 & 0
\end{bmatrix} \in \reals^{p\times m}$ and for some positive constant $c''_1 \leq \frac{\gamma\wedge 2-1}{2} \wedge \sqrt{\frac{c_2'}{6\kappa_+^2}}$, let
\[
\Theta_1 = \sth{A^{(i)} = \frac{\gamma d}{2}M_0 +  c''_1 \sqrt{a}\, {W}^{(i)}:   i=1,\dots, N}.
\]
Note that 
each $A^{(i)}$ has $s$ nonzero rows and $r$ nonzero columns.
Moreover, for $i\in [N]$, and $j\in [r]$
\begin{align*}
\left|\sigma_j(A^{(i)}) - \sigma_j(\frac{\gamma d}{2}M_0) \right| 
\leq \opnorm{A^{(i)} - \frac{\gamma d}{2} M_0} 
= c_1''\sqrt{a} \opnorm{{W}^{(i)}} 
\leq c_1''\sqrt{a}\fnorm{{W}^{(i)}} \leq \frac{\gamma\wedge 2 - 1}{2}d.
\end{align*}
Here, the second last inequality holds since
$\opnorm{{W}^{(i)}} \leq \fnorm{{W}^{(i)}}\leq 1$,
and the last inequality holds since $c_1''\leq \frac{\gamma\wedge 2-1}{2}$ and $\sqrt{a}\leq d$.
Since $\sigma_j(\frac{\gamma d}{2}M_0) = \frac{\gamma d}{2}$ for all $j\in [r]$, and so $\sigma_j(A^{(i)})\in [d, \gamma d]$ for all $j\in [r]$ and $i\in [N]$.
Thus, $\Theta_1 \subset \Theta(s,r,d,\gamma)$.

For any $i\neq j$, $D(P_{A^{(i)}}|| P_{A^{(j)}}) = \frac{1}{2}\fnorm{XA^{(i)} - XA^{(j)}}^2 \leq (c_1'' \kappa_+)^2 a$, and
\begin{align*}
\sqnorm{A^{(i)} - A^{(j)}} \geq c_1'' c_1' \sqrt{a}\,  r^{1/q - 1/2}.
\end{align*}
Hence, for $c_1 = c_1' c_1''$, $c_2 = c_2'/2$ and $c_3 = (c_1'' \kappa_+)^2$, $\dKL(\calF_0) \leq c_3 a$ and
\begin{align*}
\log\calM(\Theta_1, \sqnorm{\cdot}, c_1\sqrt{a}\, r^{1/q-1/2}) \geq c_2' (s-r)\log\frac{\eexp (p-r)}{s-r} \geq c_2 s\log\frac{\eexp p}{s}.
\end{align*}
Here, the second inequality holds since $s\geq 2r$ and $\frac{p-r}{s-r}\geq \frac{p}{s}$.
Moreover, by our choice of $c_3$, it is guaranteed that $c_3\leq c_2/3$.
This completes the proof of the first claim.

To establish the second claim, we note that \prettyref{lmm:scatter} continues to hold if we replace $s$ with $k$ and $W$ with $W'$.
Thus, we could essentially repeat the foregoing arguments to obtain the second claim. 
This completes the proof.
\end{proof}

\begin{proof}[Proof of \prettyref{thm:lower}]
Throughout the proof, let $c > 0$ denote a generic constant that depends only on $\gamma$ and $\kappa_+$, though its actual value might vary at different occurrences.
Note that we only need to prove the lower bounds for $\sigma = 1$, and the case of $\sigma \neq 1$ follows directly from standard scaling argument.

First, by restricting the nonzero entries of any matrix in $\Theta(s,k,r,d,\gamma)$ to the top left $s\times r$ (or $r\times k$) corner, we obtain a minimax lower bound by applying \prettyref{lmm:oracle-lowbd}, \ie,
for $\Theta = \Theta(s,r,d,\gamma)$ and any $q\in [1,2]$,
\begin{align}
\label{eq:oracle-lowbd}
\inf_{\wh{A}} \sup_{\Theta}\Expect\sqnorm{\wh{A} - A}^2 \geq c(r^{2/q-1} d^2) \wedge (r^{2/q} (s+k)).
\end{align}
Here, we have used the fact that for any $a,b,c>0$, 
\begin{align}
	\label{eq:lowbd-combine}
(a\wedge b)\vee (a\wedge c) = a\wedge (b\vee c) \asymp a\wedge (b+c).
\end{align}


Next, by \prettyref{prop:fano}, \prettyref{lmm:packing} and \eqref{eq:lowbd-combine},
we obtain
\begin{align}
\label{eq:comb-lowbd}
\inf_{\wh{A}} \sup_{\Theta}\Expect\sqnorm{\wh{A} - A}^2 \geq c (\sqrt{a} \, r^{1/q-1/2})^2 = c (r^{2/q-1} d^2)\wedge \pth{r^{2/q-1} \pth{s\log\frac{\eexp p}{s} + k\log\frac{\eexp m}{k}}}.
\end{align}
Thus, the minimax risk is lower bounded by the maximum of the lower bounds in \eqref{eq:oracle-lowbd} and \eqref{eq:comb-lowbd}.
Applying \eqref{eq:lowbd-combine} again, we complete the proof. 
\end{proof}

\subsection{A Theorem on Group Lasso}
\label{sec:glasso}

\begin{theorem}\label{thm:glasso} Consider the linear model $W=XB+Z$, 
where $W$ is an $n\times r$ response matrix, $X$ is an $n\times p$ design matrix, $B$ is a $p\times r$ coefficient matrix with $s$-sparse row support for some $s\geq 1$, and $Z$ is an $n\times r$ error matrix.
Let
\bes
\Bhat=\argmin_{B\in \R^{p\times r}} \|W -  XB\|_F^2/2+\lam\|B\|_{2,1},
\ees
with a given penalty level $\lam$. 
Let Condition \ref{cond-speigen} hold with an absolute constant $K>1$ and positive constants $s_*, c_*$ satisfying \eqref{eq:sparse-eig}.
%
\\
(i) If $2\|X_{*j}'(W-XB)\|_F\le \lam$ for all $j$,
then it holds that
\bel{eq:thmglas}
\|\Bhat-B\|_F\le \frac{3(1+(4c_*)^{-1})}{\kappa^2_-(s_*)}\sqrt{s}\lam.
\eel
(ii) Assume the error matrix $Z$ has iid $N(0,\sigma^2)$ entries. For any given $\eta\in (0,1)$, if we set
\bes
\lam\ge 2\sigma\max_j\|X_{*j}\|(\sqrt{r}+\sqrt{2\log(p/\eta)}),
\ees
then (\ref{eq:thmglas}) holds with probability at least $1-\eta$.

\end{theorem}

\begin{proof}[Proof of \prettyref{thm:glasso}.]
We may rewrite the minimization problem in a vectorized version as follows
\bes
\min_{B\in \R^{p\times r}} \|\text{vec}(W) - (I_r \otimes X) \text{vec}(B)\|_2^2/2+\lam\|B\|_{2,1},
\ees
where $\text{vec}$ is usual vectorization operator and $\otimes$ is the Kronecker product as defined in \cite[Section 2.2]{muirhead}. In this case, the rows of $B$ form natural groups which are all of size $r$ and $\text{vec}(B)$ satisfies the $(s,rs)$ strong group-sparsity as defined in \cite{Huang10}. 

We are to prove the desired result by invoking Lemma D.4 of \cite{Huang10}.
To this end, we first verify that the two conditions of the lemma is satisfied.
Note that the penalty level in \cite{Huang10} corresponds to $2\lam/(nr)$ in our notion, $X_{G_j}$ corresponds to $X_{*j}$, and the sparse eigenvalues $\rho_{+}(G_j)$ and $\rho_{\pm}(rs)$ are identified as 
\bes
\rho_{+}(G_j)=\|X_{*j}\|^2/(nr),\quad
\rho_{\pm}(rs)=\kappa^2_\pm(s)/(nr).
\ees
Let $\ell= s_* -s-1$ and $\lam_-^2 = \min \{k\lam^2: kr\ge \ell r+1, k\in \integers^+\}=(\ell+1)\lam^2$. The conditions of \citet[Lemma D.4]{Huang10} can be rewritten in our notation as
\bel{eq:condlemD4} 
2\|X_{*j}'(W-XB)\|_F \le \lam 
\qquad \text{and} \qquad
\frac{\tkappa^2_+(s_*,s_*-s)}{\kappa_-^2(s_*)} \le \sqrt{\frac{\ell+1}{s}},
\eel
where 
$\tkappa^2_+(s_*,s_*-s)=\sqrt{(\kappa^2_+(s_*)-\kappa^2_-(2s_*-s))(\kappa^2_+(s_*-s)-\kappa^2_-(2s_*-s))}$.

Since by \prettyref{def:riesz}, $\kappa^2_-(s)\le \kappa^2_-(t) \le \kappa^2_+(t)\le \kappa^2_+(s),\ \forall t\leq s$, we obtain
\bes
\tkappa^2_+(s_*,s_*-s)\le \kappa^2_+(s_*)-\kappa^2_-(2s_*).
\ees
Thus, the conditions in (\ref{eq:condlemD4}) are satisfied under the assumption of \prettyref{thm:glasso}. 
Then the conclusion of \citet[Lemma D.4]{Huang10} 
leads to
\bes
\|\Bhat-B\|_F\le \frac{3}{\kappa^2_-(s_*)}(1+1.5\sqrt{s/(\ell+1)})\sqrt{s}\lam
\le \frac{3(1+(4c_*)^{-1})}{\kappa^2_-(s_*)}\sqrt{s}\lam.
\ees
This completes the proof of part (i).

Turning to part (ii), we need to upper bound $2\|X_{*j}'(W-XB)\|_F$. Since $X_{*j}'(W-XB)$ is a vector of length $r$ with iid $N(0,\sigma^2\|X_{*j}\|^2)$ entries, it follows from \citet[Eq.(4.3)]{Laurent00}
that with probability $1-\eta/p$,
\bes
\|X_{*j}'(W-XB)\|_F^2  &\le& \sigma^2\|X_{*j}\|^2
(r+2\sqrt{r\log(p/\eta)} + 2\log(p/\eta))
\cr
&\le& \sigma^2\|X_{*j}\|^2(\sqrt{r}+\sqrt{2\log(p/\eta)})^2.
\ees
With probability at least $1-\eta$, we have $2\|X_{*j}'(W-XB)\|_F\le \lam$ for all $j$ and thus (\ref{eq:thmglas}) holds. 
\end{proof}

\bibliographystyle{chicago}
\bibliography{spca}

\end{document}